\documentclass[12pt]{article}
 
\usepackage[a4paper,bindingoffset=0.2in,%
            left=0.9in,right=0.9in,top=1in,bottom=1in,%
            footskip=.25in]{geometry}

\usepackage[font=small,skip=0pt]{caption}

\usepackage{pgfplots} 
\usetikzlibrary{arrows,matrix,positioning}
\usepackage[normalem]{ulem}
\usepgfplotslibrary{fillbetween}
\usepackage{hyperref}
\usepackage{amssymb}
\usepackage{amsfonts} 
\usepackage{amsmath} 
\usepackage{url}
\usepackage{cite}
\usepackage{amsthm}
\theoremstyle{plain}
\newtheorem{theorem}{Theorem}
\newtheorem{corollary}{Corollary}
\newtheorem{lemma}{Lemma}
\newtheorem{proposition}{Proposition}
\theoremstyle{definition}
\newtheorem{definition}{Definition}

\theoremstyle{remark}
\newtheorem{remark}{Remark}
\usepackage{algorithm,algcompatible,lipsum}
\usepackage{color}
\usepackage{tcolorbox}
\usepackage{todonotes}
\newcommand{\mbf}{\mathbf}
\newcommand{\mbb}{\mathbb}
\newcommand{\mcl}{\mathcal}
\newcommand{\mrm}{\mathrm}
\newcommand{\bld}{\boldsymbol}
\usepackage{xstring}
\def\alphabet{abcdefghijklmnopqrstuvwxyzABCDEFGHIJKLMNOPQRST123456789}
\renewcommand{\vec}[1]{
\IfSubStr{\alphabet}{#1}{
\ensuremath{\mathbf{\MakeLowercase{#1}}}
}{
\ensuremath{\boldsymbol{\MakeLowercase{#1}}}
}
}
\newcommand{\mat}[1]{
\IfSubStr{\alphabet}{#1}{
\ensuremath{\mathbf{\MakeUppercase{#1}}}
}{
\ensuremath{\boldsymbol{\MakeUppercase{#1}}}
}
}

\def\R{\mathbb R}
\def\C{\mathbb C}
\def\N{\mathbb N}

\newcommand*{\inner}[2]{\left\langle#1,#2\right\rangle}
\newcommand*{\norm}[1]{\left\|#1\right\|}
\newcommand*{\paran}[1]{\left(#1\right)}
\newcommand*{\bracket}[1]{\left\{#1\right\}}
\newcommand*{\card}[1]{\left|#1\right|}
\def\defeq{:=}
\def\d{\mathrm{d}} 
\def\st{\text{ subject to }}

\def\S{\mathbb{S}} 
\def\SO{\mathrm{SO}} 
\newcommand*{\Y}[2]{\mathrm{Y}_{#1}^{#2}} 
\newcommand*{\D}[3]{\mathrm{D}_{#1}^{#2,#3}} 
\newcommand*{\Wd}[3]{\mathrm{d}_{#1}^{#2,#3}} 
\def\i{\mathrm{i}}

\usepackage{xcolor}

\usepackage{acronym}
\acrodef{CS}{Compressed Sensing}
\acrodef{BP}{Basis Pursuit}
\acrodef{OMP}{Orthogonal Matching Pursuit}
\acrodef{AMP}{Approximate Message Passing}
\acrodef{RIP}{Restricted Isometry Property}
\acrodef{BOS}{Bounded Orthonormal System}
\acrodef{IGRF}{International Geomagnetic Reference Field}
\begin{document}

\title{Sensing Matrix Design and Sparse Recovery on the Sphere and the Rotation Group}

\author{Arya Bangun, Arash Behboodi, and Rudolf Mathar\footnote{
Institute for Theoretical Information Technology, RWTH Aachen University 
}
}
\date{}
\maketitle 
\begin{abstract}
In this paper, {the goal is to design deterministic sampling patterns on the sphere and the rotation group} and, thereby, construct sensing matrices for sparse recovery of band-limited functions. It is first shown that random sensing matrices, which consists of random samples of Wigner D-functions,  satisfy the \acf{RIP} with proper preconditioning and can be used for sparse recovery on the rotation group.  The mutual coherence, however, is used to assess the performance of
deterministic and regular sensing matrices. 
 We show that many of widely used regular sampling patterns yield sensing matrices with the worst possible mutual coherence, and therefore are undesirable for sparse recovery. Using tools from angular momentum analysis in quantum mechanics, we provide a new expression for the mutual coherence, which encourages the use of regular elevation samples. We construct low coherence deterministic matrices by fixing the regular samples on the elevation and minimizing the mutual coherence over the azimuth-polarization choice. It is shown that once the elevation sampling is fixed, the mutual coherence has a lower bound that depends only on the elevation samples. This lower bound, however, can be achieved for spherical harmonics, which leads to new sensing matrices with better coherence than other representative regular sampling patterns. This is reflected as well in our numerical experiments where our proposed sampling patterns perfectly match the phase transition of random sampling patterns.
\end{abstract}

\section{Introduction}
In many applications, where the goal is to recover a sparse signal from the fewest linear  measurements, the measurement process cannot be freely chosen. That is, in the corresponding linear inverse problem, the sensing matrix has a specific structure. A central question, therefore, is to design the sensing matrix under these additional restrictions.

For general sensing matrices, the pioneering works of compressed sensing \cite{candes2005decoding,candes2006robust,candes_near-optimal_2006} followed by overwhelming subsequent researches established recovery guarantees for various random matrices including subgaussian random matrices. These random matrices are shown to satisfy, with high probability, the \acf{RIP}, which is a sufficient condition for noise-robust sparse recovery. Many efficient algorithms such as \ac{BP} can provably recover the original signal from these measurements 
(see \cite{foucart2013mathematical} for an exhaustive treatment of the subject).

 In contrast to pure random matrix designs, in many applications, the sensing medium imposes additional structures on sensing matrices. Notable examples are sensing matrices that are obtained from sampling  functions in finite-dimensional function spaces. The sensing matrix entries in these applications are samples of orthonormal basis functions of the ambient space. Fourier matrices \cite{rudelson_sparse_2008}, matrices from trigonometric polynomials \cite{rauhut2007random}, orthogonal polynomials \cite{rauhut2012sparse,szeg1939orthogonal} and spherical harmonics \cite{rauhut2011sparse,burq2012weighted} are some examples of these matrices.  Fortunately, when the orthonormal functions are uniformly bounded, also called  \acp{BOS}, a similar recovery guarantee can be obtained. If the samples are taken randomly from a certain probability measure, \ac{BOS} matrices are proven to satisfy the \ac{RIP} property \cite[Chapter 12]{foucart2013mathematical}. If the orthonormal functions are uniformly bounded by $K$, the required number of measurements scales with $K^2$.

This randomness in the measurement process, however, is inadmissible in many applications, for instance when the measurement process involves movements of mechanical devices. Random measurements require arbitrary movements that are possibly harmful to the measurement device. In these applications, the measurement process should be designed by considering the physical characteristics of the measurement device. An example, which is the main motivation of the current work, is the antenna measurement application. The samples in antenna measurements are taken using a robotic arm or which samples of a smooth trajectory are preferred over random samples. Therefore, regular sampling patterns like equiangular patterns are widely used for the measurement process. 
The desired sensing matrices should be both structured, since it involves samples of orthonormal functions, and deterministic, which should bring about regular sampling patterns. In this paper, our goal is to address these requirements step by step for sparse recovery in the space of band-limited square-integrable functions over the sphere $\S^2$ and the rotation group $\SO(3)$. These functions appear in a wide range of applications such as antenna measurements \cite{CoHeKoBeMa0}, geophysics \cite{thebault2015international}, spherical microphone arrays  \cite{rafaely_analysis_2005}, and astrophysics \cite{jarosik_seven-year_2011}.

Consider random measurements first. The orthonormal functions over $\S^2$ and $\SO(3)$ are spherical harmonics and Wigner D-functions, random samples of which constitutes the entries of the sensing matrix. The upper bound of these functions, $K$, is a function of the ambient dimension $N$. {For example, the bound $K$ for band-limited spherical harmonics with bandwidth $B$ is equal to $\sqrt{({2B-1})/{4\pi} }$. The number of band-limited functions $N$ is equal to ${B^2}$, which implies that 
$K=\sqrt{({2\sqrt{N}-1})/{4\pi} }$.} When this is plugged in the recovery guarantees for \acp{BOS}, it would imply that the number of measurements should scale badly with the dimension $N$. This bound is useless for sparse recovery analysis. Rauhut and Ward used a preconditioning technique in \cite{rauhut2011sparse} and improved the dependence to $N^{1/4}$. Burq et al. improved this further to $N^{1/6}$ in  \cite{burq2012weighted}. These results, however, do not directly generalize to Wigner D-functions. 

As soon as we move to deterministic sampling patterns, the \ac{RIP} cannot be used to appraise the sparse recovery capability of the sensing matrices. It is computationally hard to certify that a certain matrix satisfies \ac{RIP} \cite{tillmann2014computational,bandeira2013certifying}. A common metric for deterministic sensing matrices is the mutual coherence. It is defined as the maximum of the absolute value of normalized inner products between columns of the sensing matrix. Unlike \ac{RIP}, the mutual coherence can be numerically evaluated for a given matrix, and therefore it is a computable figure of merit for sparse recovery. The mutual coherence of a matrix can also be used to provide recovery guarantees, although it leads to a suboptimal dependence on the sparsity order. In general, sensing matrices with low mutual coherence tend to have better sparse recovery performance. Therefore, constructing a sensing matrix with low mutual coherence has been widely investigated in recent years because of its extensive application in many different areas, from coding theory and communication   \cite{strohmer2003grassmannian,love2003grassmannian}, compressed sensing \cite{elad2007optimized,candes2006robust,candes_near-optimal_2006,cande2008introduction,tropp2004greed},  quantum measurement \cite{eldar2002optimal}, and machine learning \cite{drineas2012fast,mohri2011can}. The mutual coherence is lower bounded by the Welch bound, obtained in the context of correlation measurements of different signals \cite{welch1974lower}.
The lower bound is tight  and can be achieved by equiangular and tight frames \cite{strohmer2003grassmannian}. A similar result for structured matrices is not known to the best of our knowledge. 

{\subsection{Related Works}}
{Deterministic sampling patterns on the sphere $\S^2$ have been studied extensively in the context of Shannon-Nyquist sampling for the reconstruction of band-limited functions (see \cite{mcewen_novel_2011,hagai_generalized_2012,mcewen2013sparse} and references therein). As mentioned by McEwen and Wiaux in \cite{mcewen_novel_2011}, some of these techniques can be used to enhance the performance of compressed sensing methods. Equiangular sampling patterns are often the standard in these applications. In these works, to represent a band-limited function with a bandwidth of $B$, the number of samples should scale as  $\mcl{O}(B^2)$, which is linear in the ambient dimension. For high bandwidth signals, it implies a long measurement time. Compressed sensing approach can circumvent this issue by leveraging the sparsity structure in the signal. We will see, however, that equiangular sampling patterns are not good choices for compressed sensing.} 

{\ac{CS} over sphere has been considered in a few works. Deterministic sensing matrix design from spherical harmonics was considered in \cite{alem_sparse_2012} where spiral sampling points are used as the basis for the design. They show that the proposed sampling points outperform equiangular sampling. However, those works emphasize numerical evaluations of the sparse recovery without analyzing the structure of sensing matrices, discussing the achievable coherence bounds or providing design guidelines. The authors in  \cite{alem_sparse_2014} considered probabilistic \ac{CS} to provide recovery guarantees without using RIP. Relying on preconditioning approaches of  \cite{burq2012weighted}, the approach provided a probabilistic recovery guarantees for which the number of measurements depends on $N^{1/6}$ for spherical harmonics. These results usually rely on some conditions on the sensing matrix (see \cite[Chapter 14]{foucart2013mathematical}). Besides, the result does not hold uniformly over all vectors. In many applications, it is difficult to conduct measurements using random sampling patterns. To the best of our knowledge, this work is the first to consider functions over the rotation group and provide mutual coherence-based guidelines for designing deterministic sensing matrices of spherical harmonics and Wigner D-functions.
}

\subsection{Summary of Contributions}
{In this paper, we consider the problem of sensing matrix design and evaluation for sparse recovery of band-limited functions on the sphere $\S^2$ and the rotation group $\SO(3)$. Although there are some works on sensing matrix design over the sphere, the problem is almost unexplored for the rotation group. One of the main contributions of this paper is to study sensing matrix design for compressed sensing over the rotation group. The sensing matrix design boils down to finding $m$ sampling points $(\theta,\phi) \in \S^2$ and 
$(\theta,\phi,\chi) \in \SO(3)$, where $\theta\in[0,\pi],\phi\in[0,2\pi)$ and $\chi\in[0,2\pi)$.
 For spherical harmonics, certain random sampling patterns can be provably used for sparse recovery \cite{rauhut2012sparse}. In Section \ref{Sect:RIP}, we prove that it is also possible to find a pair of points on the rotation group with guaranteed sparse recovery. Specifically, sparse band-limited signals over the rotation group $\SO(3)$ can be uniquely recovered from certain random sampling patterns by solving a convex optimization problem. The proof follows from the \ac{RIP} property of the sensing matrix after preconditioning, which is based on some inequalities for Jacobi polynomials. The required number of samples scale with the ambient dimension as $N^{1/6}$. The recovery algorithm is robust to noise and stable to model inaccuracies. These results show that it is possible to find sampling patterns with recovery guarantees on the sphere and the rotation group. The focus of our paper, however, is on deterministic sampling pattern design. This paper discusses  for the first time compressed sensing over the rotation group. We provide new tools, guidelines and designs for the problem of deterministic sensing matrix design over the sphere and the rotation group. The main contributions of our paper are as follows. }

\begin{itemize}
\item {Adopting the mutual coherence as the figure of merit from Section \ref{Sect:coherence}, we show that certain regular deterministic sampling patterns over the sphere and the rotation group with symmetric structures over $\phi$ and $\chi$ have maximum coherence and, therefore, are not good for sparse recovery. These patterns include many of sampling patterns that are currently widely used in applications, including equiangular sampling patterns.}

\item The mutual coherence is determined by the inner products of vectors of samples  of spherical harmonics and Wigner D-functions. We show in Section \ref{Sect:coherence} that the product of two functions can be seen as the total angular momentum of a composite quantum system. Borrowing this insight from quantum mechanics, the product can be decomposed into a sum of single spherical harmonics and Wigner D-functions using Wigner 3j symbols. To the best of our knowledge, this decomposition is used for the first time for coherence analysis in compressed sensing. We use the above decomposition to derive regular sampling patterns that lead to mutually orthogonal, and therefore incoherent, columns in the sensing matrix. 

\item {In Section \ref{Sect:equispacedSampling}, we propose equispaced sampling patterns on $\theta$, which also leads to incoherent columns. We show that once the sampling points on $\theta$ is fixed the mutual coherence is automatically lower bounded independent of the choice of $\phi$'s and $\chi$'s. It is, however, shown that the lower bound can be achieved for spherical harmonics with our newly proposed sampling pattern. The new sampling pattern is obtained by an algorithm that minimizes the mutual coherence using pattern search algorithm. Although the lower bound cannot be achieved for Wigner D-functions using this method, the mutual coherence of our proposed pattern is still superior to the representative regular sampling patterns.} 
\item {Our phase transition diagrams in Section \ref{Sect:Experiments} suggest that our proposed sampling pattern not only outperforms the representative regular patterns but also matches perfectly random sampling patterns. We demonstrate the benefit of our sampling pattern in some potential applications. These applications include spherical near-field antenna measurements as well as the reconstruction of the earth's magnetic field. We show that, for both cases the required number of samples can be significantly reduced.}
\end{itemize}

The codes used in this paper are available below:
\begin{center}
 {\fontfamily{lmss}\selectfont \url{github.com/bangunarya/}samplingsphere}
\end{center}

\subsection{Notation}
The vectors are denoted by bold small-cap letters. Define $\N\defeq\{1,2,\dots\}$ and $\N_0\defeq \N\cup\{0\}$. Throughout the paper, $a\lesssim b$ means that there is a universal constant $C$ such that $a\leq C b$. Similar convention is used for $a\gtrsim b$. $f(\mathbf x)$ for a function $f:\mathbb R\to\mathbb R$ is the element-wise application of $f$ to the vector $\mathbf x$. $\overline{x}$ is the conjugate of $x$.

\section{Definitions and Backgrounds}\label{Sect:definition}
In this section, we introduce briefly the preliminaries of signal processing over the sphere and the rotation group as well as the problem formulation. The central problem of this work is the recovery of band-limited functions defined on the sphere and the rotation group. We need, therefore, to introduce Fourier analysis for these spaces of functions. 

\subsection{Spherical Harmonics and Wigner D-functions}
Consider the Hilbert space of square-integrable functions $f(\cdot)$ on the sphere ${\S^2}$ denoted by $L^2(\S^2)$. Each element of $\S^2$ is represented by two numbers $\theta\in[0,\pi]$ and $\phi\in[0,2\pi)$. The variables $\theta$ and $\phi$ are called the elevation and the azimuth. 
 The inner product of $f,g\in L^2(\S^2)$ is defined by
\[
\inner{f}g\defeq\int_{\S^2}f(\theta,\phi)\overline{g(\theta,\phi)}\d\nu(\theta,\phi),
\]
where $\d\nu(\theta,\phi)\defeq\sin \theta \d\theta \d\phi$ is the uniform measure on the sphere.
Spherical harmonics are basis functions for the space of functions in $L^2(\S^2)$.  Denoted by $\Y lk(\theta,\phi)$ for degree $l\in\N_0$ and order $k\in\{-l,\dots,l\}$,  they are defined over the sphere ${\S^2}$ as follows:
\begin{equation} \label{SH}
\Y lk (\theta,\phi) \defeq  N^k_l P_l^{k}(\cos \theta) e^{\i k\phi},
\end{equation}
where $P_{l}^k (\cos \theta)$ is the associated Legendre polynomials defined by
\[
P_l^k(x)\defeq \frac{(-1)^k}{2^ll!}(1-x^2)^{k/2}\frac{\d^{k+l}}{\d x^{k+l}}(x^2-1)^l.
\]
The term $N^k_l\defeq\sqrt{\frac{2l+1}{4\pi} \frac{(l-k)!}{(l+k)!}}$ is a normalization factor. It ensures that the function $\Y lk$ has unit $L_2$-norm.
Spherical harmonics are orthonormal with respect to the uniform measure on the sphere  $\d\nu=\sin \theta \d\theta \d\phi$, i.e.,
\begin{equation}
\begin{aligned}
& \int_{0}^{2\pi}\int_{0}^{\pi} \mathrm Y_l^{k}(\theta,\phi) \overline{\mathrm Y_{l'}^{k'} (\theta,\phi)} \sin \theta \mathrm{d}\theta \mathrm{d}\phi = \delta_{ll'} \delta_{kk'}
\end{aligned}
\end{equation}
where $\delta_{ll'}$ is the Kronecker delta. The function $\overline{\Y lk}$ is the conjugate of $\Y lk$ and satisfies:
\[
\overline{\Y lk(\theta,\phi)}=(-1)^k\Y{l}{-k}(\theta,\phi).
\]
For any function $f\in L^2(\S^2)$, the unique expansion
\begin{equation}
f(\theta,\phi)=\sum_{l=0}^{\infty}\sum_{k=-l}^{l} \hat{f}_l^{k} \,\mathrm{Y}_l^{k}(\theta,\phi),
 \label{eq:L2funcS2}
\end{equation} 
where
\begin{equation}
\hat{f}_l^{k}=\int_{0}^{2\pi}\int_{0}^{\pi} f(\theta,\phi) \,\overline{\mathrm{Y}_l^{k}(\theta,\phi)} \sin\theta \mathrm{d}\theta \mathrm{d}\phi.
\end{equation} 
is called the ${\S^2}$-Fourier expansion of $f$ with Fourier coefficients $\hat{f}_l^{k}$. 

The space of all rotations of the sphere $\S^2$ is a group called the rotation group and is denoted by $\SO(3)$. Each element of $\SO(3)$ can be represented by  three rotation angles $\phi\in[0,2\pi)$, $\theta\in[0,\pi]$, and $\chi\in[0,2\pi)$. In this work, we call the angle $\chi$ the polarization. The Hilbert space of square integrable functions on $\SO(3)$, denoted by $L^2(\SO(3))$, is endowed with an inner product, which is defined for two functions $f,g\in\SO(3)$ by
\[
\inner{f}g\defeq \int_{\SO(3)} f(\theta,\phi,\chi)\overline{g(\theta,\phi,\chi)}\d\nu(\theta,\phi,\chi),
\] 
where $\d\nu(\theta,\phi,\chi)\defeq\sin \theta \d\theta \d\phi\d\chi$.  Wigner D-functions are an orthonormal basis for the Hilbert space $L^2(\SO(3))$. Denoted by $\D l{k}{n}(\theta,\phi,\chi)$ with degree $l\in\N_0$ and orders $k,n\in\{-l,\dots,l\}$, they are defined by 
\begin{equation}
\D l{k}{n}(\theta,\phi,\chi)= N_l e^{-\i k\phi} \mathrm{d}_l^{k,n}(\cos \theta)  e^{-\i n\chi} 
 \label{def:WigD}
\end{equation}
where $N_l=\sqrt{\frac{2l+1}{8\pi^2}}$ is the normalization factor to guarantee that Wigner D-functions are unit norm.
The function $\Wd l{k}{n}(\cos \theta)$ is the Wigner d-function of oder $l$ and degrees $k,n$ defined by:
\begin{equation}
\Wd l{k}{n}(\cos \theta)= \omega \sqrt{\gamma} \sin^{\xi} \bigg(\frac{\theta}{2}\bigg)\cos^{\lambda}\bigg(\frac{\theta}{2}\bigg) P_{\alpha}^{(\xi,\lambda)}(\cos \theta)
\end{equation}
where $\gamma=\frac{\alpha!(\alpha + \xi + \lambda)!}{(\alpha+\xi)!(\alpha+\lambda)!}$, $\xi=\card{k-n}$, $\lambda=\card{k+n}$, $\alpha=l-\big(\frac{\xi+\lambda}{2}\big)$ and 
\begin{equation*}
 \omega= \begin{cases} 
        1  & \text{if } n\geq k  \\
        (-1)^{n-k} & \text{if } n<k 
          \end{cases}.
\end{equation*}
{The function $P_{\alpha}^{(\xi,\lambda)}$ is the Jacobi polynomial defined by
 \begin{align*}
P_{\alpha}^{(\xi,\lambda)}(x) &\defeq \frac{(-1)^{\alpha}}{2^{\alpha} \alpha !}(1-x)^{-\xi} (1+x)^{-\lambda} \times \frac{\d^{\alpha}}{\d x^{\alpha}}\left( (1-x)^{\xi}(1+x)^{\lambda}(1-x^2)^{\alpha} \right).
\end{align*}}
The orthonormal property of Wigner D-functions writes as:
\begin{equation}
\begin{aligned}
 \int_{0}^{2\pi}\int_{0}^{2\pi}\int_{0}^{\pi}& \mathrm D_l^{k,n}(\theta,\phi,\chi) \overline{\mathrm D_{l'}^{k',n'} (\theta,\phi,\chi)} \sin \theta \mathrm{d}\theta \mathrm{d}\phi \mathrm{d}\chi =\delta_{ll'} \delta_{kk'} \delta_{nn'}.
\end{aligned}
\end{equation}
The conjugate of $\D lkn$ satisfies \cite[eq. 7.134]{kennedy_hilbert_2013}
\[
\overline{\D lkn(\theta,\phi,\chi)}=(-1)^{k-n}\D l{-k}{-n}(\theta,\phi,\chi).
\]

The $\SO(3)$-Fourier expansion of the function $g \in L^2(\SO(3))$ is defined by  
\begin{equation}
g(\theta,\phi,\chi)=\sum_{l=0}^{\infty}\sum_{k=-l}^{l}\sum_{n=-l}^{l} \hat{g}_l^{k,n} \,\D l{k}{n}(\theta,\phi,\chi),
\label{wfexp}
\end{equation} 
with Fourier coefficients $\hat{g}_l^{k,n}$ are obtained by
\begin{equation}
\hat{g}_l^{k,n}=\int_{0}^{2\pi}\int_{0}^{2\pi}\int_{0}^{\pi} g(\theta,\phi,\chi) \,\overline{\D l{k}{n}(\theta,\phi,\chi)} \sin\theta \d\theta \d\phi \d\chi.
\end{equation} 
An interested reader can refer to the book \cite{wigner2012group} 
for more information on Wigner D-functions and $\SO(3)$.

\begin{remark}
If the order $n$ is set to zero, we get spherical harmonics.  
The Wigner D-functions $\D l{k}{0}$ for $n=0$ are related to spherical harmonics $\Y{l}k$ as 
\begin{equation}\label{gener_SH}
\D l{-k}{0}(\theta,\phi,0) = (-1)^k \sqrt{\frac{1}{2\pi}}\Y {l}k(\theta,\phi).
\end{equation}
 
\end{remark}
\subsection{Sparse Expansions of Band-limited Functions}
In this work, we are interested in band-limited functions inside $L^2(\S^2)$. A function $f \in L^2(\S^2)$ is band-limited with bandwidth $B$ if it is expressed in terms of spherical harmonics of degree less than $B$:
\[
 f(\theta,\phi)=\sum_{l=0}^{B-1}\sum_{k=-l}^{l} \hat{f}_l^{k} \,\Y l{k}(\theta,\phi).
\]
The space of band-limited functions with the degree less than $B$ is a subspace of $L^2(\S^2)$ of dimension $N=B^2$. Every band-limited function $f$, therefore, is fully determined by the vector of $N$  Fourier coefficients $\vec{f}=(\hat{f}_l^{k})_{0\leq l<B}$.

We can define similarly the notion of band-limited functions on $\SO(3)$. A function $g \in L^2(\SO(3))$ is band-limited with bandwidth $B$ if it is expressed in terms of Wigner D-functions of degree less than $B$:
\[
 g(\theta,\phi,\chi)=\sum_{l=0}^{B-1}\sum_{k=-l}^{l}\sum_{n=-l}^{l} \hat{g}_l^{k,n} \,\D l{k}{n}(\theta,\phi,\chi).
\]
The space of band-limited functions with the degree less than $B$ is a subspace of $L^2(\SO(3))$ of dimension $N=\frac{B(2B-1)(2B+1)}{3}$ where each function is completely determined by the vector of Fourier coefficients, $\mathbf{g}=(\hat{g}_l^{k,n})_{0\leq l< B}$.

A band-limited function, whether in  $L^2(\S^2)$ or in $L^2(\SO(3))$, is said to be $s$-sparse if {the vector of its Fourier coefficient $\vec{x}$, i.e.,  $\vec{x} = \vec{f}$ or $\vec{x} = \vec{g}$,} has at most $s$ non-zero entries. This is stated in terms of the $\ell_0$-norm\footnote{
The
 $\ell_0$-norm of a vector $\vec{x}\in\C^n$ is defined by:
\[
\norm{\vec{x}}_0\defeq \sum_{i=1}^n\mrm{1}(x_i\neq 0),
\]
where $\mrm{1}(\cdot)$ is the identity function. Needless to say that $\ell_0$-norm is called a norm just as a convention. It is, indeed, not a norm.
}
as $\norm{\vec{x}}_0\leq s$. For the general non-sparse vector of coefficients $\vec{x}$, either in $L^2(\S^2)$ or in $L^2(\SO(3))$, the best $s$-sparse approximation error of $\vec x$ is defined by:
\[
\sigma_s(\vec{x})_p=\min_{\vec{z}\in\C^N:\norm{\vec z}_0\leq s}\norm{\vec{z}-\vec{x}}_p.
\]
In many applications, the signals are approximately sparse or compressible, that is, the $s$-sparse approximation error decreases rapidly as $s$ increases.

\subsection{Linear Inverse Problems and the $\ell_1$-minimization}

Consider a band-limited function either in $\S^2$ or $\SO(3)$. The function belongs to a finite-dimensional vector space and can be represented by its Fourier coefficients. It is therefore enough to find the Fourier coefficients, a finite-dimensional vector, to specify the function. 

We want to find the Fourier coefficients of a band-limited function from noisy linear samples of the function using as few samples as possible. We focus on $\SO(3)$, which contains $\S^2$ as a special case. Consider a function $g\in L^2(\SO(3))$. We obtain $m$ noisy samples $y_p$ of the function $g$ at points $(\theta_p,\phi_p,\chi_p)$ for $p\in[m]$. The samples are given by:
\begin{align*}
 y_p&=g(\theta_p,\phi_p,\chi_p)+\eta_p\\
 &=\sum_{l=0}^{B-1}\sum_{k=-l}^{l}\sum_{n=-l}^{l} \hat{g}_l^{k,n} \,\D l{k}{n}(\theta_p,\phi_p,\chi_p)+\eta_p,
\end{align*}
where $\eta_p$ is the additive noise with $\card{\eta_p}\leq \epsilon$. The noisy samples are therefore linearly related to the coefficients $\mathbf{g}=(\hat{g}_l^{k,n})_{0\leq l< B}$ as follows:
\begin{equation}
\mbf y=\mbf A\mbf g+\bld\eta, 
\label{eq:LIP}
\end{equation}
where the sample  and the noise vectors are given by:
\[
\vec{y}=\begin{pmatrix}
  y_1\\
  \vdots\\
  y_m
\end{pmatrix},
\vec{\eta}=\begin{pmatrix}
  \eta_1\\
  \vdots\\
  \eta_m
\end{pmatrix}.
\]
The noise vector satisfies $\norm{\bld\eta}_\infty\leq \epsilon$. The matrix $\mat{A}$, called the measurement or sensing matrix, is given by:
\begin{equation}
\mat{A}=
\begin{pmatrix}
  \D 0{0}{0}(\theta_1,\phi_1,\chi_1)&\dots& \D {B-1}{B-1}{B-1}(\theta_1,\phi_1,\chi_1) \\
  \vdots\\
  \D 0{0}{0}(\theta_m,\phi_m,\chi_m)&\dots&\D {B-1}{B-1}{B-1}(\theta_m,\phi_m,\chi_m)
\end{pmatrix}.
\label{eq:sensingmatrix}
\end{equation}
{The columns of $\mat A$ consist of} $m$ different samples of Wigner D-functions, and its rows are comprised of a single sample of all Wigner D-functions of degree less than $B$. The ordering of Wigner D-functions in a row is arbitrary. The only caveat is that the vector $\mbf g\in\C^N$ of $N$ coefficients should be similarly ordered. For simplicity, we assumed that the degree and orders of the Wigner D-function in the column $q\in[N]$ are determined by three functions $l(q)$, $k(q)$ and $n(q)$. In this way, the Wigner D-function of the column $q$ is $\D {l(q)}{k(q)}{n(q)}$. The entry $q$ of $\vec{g}$ is $\hat{g}_{l{(q)}}^{k{(q)},n(q)}$, and the matrix $\mbf A$ is written as 
\begin{equation}
\mbf A=[{A}_{p,q}]_{p\in[m],q\in[N]}:\quad {A}_{p,q}=\mathrm{D}_{l{(q)}}^{k{(q)},n{(q)}}(\theta_p,\phi_p,\chi_p).
\label{eq:sensMat}
\end{equation}

The linear inverse problem is similarly defined for spherical harmonics by removing the polarization parameter from the above equation. In both cases, we are interested in finding the Fourier coefficients from a few samples.

If the vector of coefficients $\vec{g}$, or $\vec{f}$, are sparse or compressible, there are many algorithms for finding the coefficients from a number of samples $m$ {that is smaller than the dimension $N$}. 
In this paper, we use quadratically constrained basis pursuit, i.e., $\ell_1$-minimization problem to solve the problem \eqref{eq:LIP}. The focus, however, is more on different sampling patterns and their effectiveness for signal recovery. The quadratically constrained basis pursuit is defined below:
\[
\vec{g}^{\#}=\arg\min_{\vec{z}\in\C^N}\norm{\vec{z}}_1\st \norm{\mat{A}\vec{z}-\vec{y}}_2\leq \sqrt{m}\epsilon. \tag{QCBP}
\label{eq:QCBP}
\]
In the next sections, we consider various sampling patterns and their recovery guarantees.

\section{Sparse Recovery Guarantees for Random Matrices}
\label{Sect:RIP}

How should the sensing matrix $\mat A$ be chosen for the program \eqref{eq:QCBP} to find a \textit{good} approximation of compressible  coefficients vectors? The error of a good approximation is only bounded by the model and measurement inaccuracies determined by the $s$-sparse approximation error and the noise strength. Therefore, we are interested in choosing $\mat A$ such that any $s$-sparse vector can be perfectly recovered from noiseless linear measurements. This is shown to be possible in compressed sensing literature if the samples are taken randomly from a class of distributions. In most of these results, the proof amounts to showing 
\acl{RIP}, a sufficient condition for signal recovery, for the sensing matrix $\mat{A}$. The \ac{RIP} is defined below.  
\begin{definition}
A matrix $\mat A {\in \C^{m \times N}}$ satisfies the restricted isometry property of order $s$ with constant $\delta\in(0,1)$, if the following inequalities hold for all $s$-sparse vectors $\mathbf x {\in \C^N}$
$$
(1-\delta)\norm{\mathbf x}_2^2\leq \norm{\mathbf{Ax}}_2^2\leq (1+\delta)\norm{\mathbf x}_2^2. 
$$
The smallest number $\delta$, denoted by $\delta_s$, is called the restricted isometry constant of $\mathbf A$. 
\end{definition}
 
Fortunately, a general result for \acp{BOS} is available. The result is used later, and we present it for the paper to be self-contained.
 
\begin{theorem}[\ac{RIP} for \ac{BOS} {\cite[Theorem 12.31]{foucart2013mathematical}}]\label{thm:RIP_BOS}
Consider a set of bounded orthonormal  basis  {$\psi_q:\mathcal{D}\to\mathbb{C} ,q \in [N]$} that are orthonormal with respect to a probability measure $\nu$ on the measurable space $\mathcal{D}$. Consider the matrix $ \boldsymbol{\psi} \in \mathbb{C}^{m \times N}$ with entries
\begin{equation*}
{\psi}_{p,q} = \psi_q(t_p),\,\, p \in [m]\,\,,q \in [N]  
\end{equation*}
constructed with i.i.d. samples $t_p$ from the measure $\nu$. Suppose that $\sup_{q\in [N]}\norm{\psi_q}_\infty\leq {K}$. If 
\begin{equation*}
m \gtrsim \, \delta^{-2}\, {K}^2 \,s\, \log^3(s) \,\log(N)
\end{equation*}
then with probability at least $1-N^{-\gamma log^3(s)}$, the restricted isometry constant $\delta_s$ of $\frac{1}{\sqrt{m}} \psi$ satisfies $\delta_s \leq \delta$ for $\delta\in(0,1)$. The constants $C,\gamma \geq 0$ are universal.
\end{theorem}

The crucial assumption, as we will see later, is the uniform boundedness of $\psi_q(\cdot)$. 
Once the \ac{RIP} property is satisfied by a matrix, $s$-sparse vectors are recovered perfectly using the program \eqref{eq:QCBP}. \ac{RIP} property, indeed, implies the robust and stable null space property which is the necessary and sufficient condition for unique recovery (see \cite[Chapter 12]{foucart2013mathematical}). The following theorem summarizes this result.

\begin{theorem}[Sparse Recovery for \ac{RIP} Matrices {\cite[Corollary 12.34]{foucart2013mathematical}}]\label{thm:RIP_recovery}
Suppose that the matrix $ \boldsymbol{\psi} \in \mathbb{C}^{m \times N}$ has restricted isometry constant $\delta_{2s}\leq 0.4931$. Suppose that the measurements are noisy $\mathbf y=\boldsymbol{\psi}\mathbf x+\boldsymbol{\eta}$ with $\norm{\boldsymbol{\eta}}_\infty\leq \epsilon$. If $\mathbf x^\#$ is the minimizer of  
\begin{equation*}
\vec{x}^{\#}=\arg\min \norm{\mathbf z}_1\text{ subject to } \norm{\mathbf y-\boldsymbol{\psi}\mathbf x}_2\leq \epsilon,
\end{equation*}
then 
\[
\norm{\mathbf x-\mathbf x^\#}_2\lesssim C\paran{\frac{\sigma_s(\mathbf x)_1}{\sqrt{s}}+\epsilon},
\]
 where $C$ depends only on $\delta_{2s}.$
 Without noise, we have $\mathbf x=\mathbf x^\#$ for $s$-sparse vectors $\mathbf x$.
\end{theorem}

Unfortunately the results of Theorem \ref{thm:RIP_BOS} and \ref{thm:RIP_recovery} provide only weak bounds for spherical harmonics and Wigner D-functions, because these orthonormal functions are not uniformly bounded, as mentioned in \cite{rauhut2011sparse}. More precisely, see that: 
\begin{equation}
{\Y l0(0,\phi)} =\sqrt{\frac{2l+1}{4\pi}}.
\end{equation}
The value of ${\Y l0(0,\phi)}$ can be shown to be the upper bound on all spherical harmonics of degree $l$. This means that all spherical harmonics of degree less than $B$ are bounded by $\sqrt{\frac{2B-1}{4\pi}}$, and the bound is tight. Since the ambient dimension $N$ is equal to $B^2$, the uniform upper bound $K$ on spherical harmonics depends on $N$ as $K=O(\sqrt{B})$. Theorem \ref{thm:RIP_BOS}, then, yields a bound on $m$ that depends on the ambient dimension as $O(\sqrt{N})$. A more general dependence of this type appeared in the paper \cite{burq2012weighted}. This dependence {might yield vacuous bounds on the measurement numbers} for large dimensions and very sparse vectors. 

Rauhut and Ward in \cite{rauhut2011sparse} and Burq et al. in \cite{burq2012weighted} used a preconditioning technique that improves this dependence for spherical harmonics.  At the core of the preconditioning technique lies the following inequality:
\begin{equation}
\card{(\sin^2\theta\cos \theta)^{1/6} \, Y_l^k(\theta,\phi)} \lesssim (l+1)^{1/6}
\end{equation} 
Burq et al. \cite{burq2012weighted} change the probability measure defined on $\S^2$ to the measure $\mathrm d\nu=|\tan \theta|^{1/3}\mathrm  d\theta \mathrm d\phi$ and preconditioned the spherical harmonics by $(\sin^2\theta\cos \theta)^{1/6}$. Note that further normalization by a constant is needed to turn the new measure to a probability measure. The new probability measure, however, improves the dependence of $m$ on $N$ to $O(N^{1/6})$, which improves also the previous precondtioning by $(\sin \theta)^{1/2}$ proposed in \cite{rauhut2011sparse}.

The upper bound on Wigner D-functions, similarly, depends on $N$. In particular, see from the equality \ref{gener_SH}, that the upper bound $K$ is also $O(\sqrt{B})$. Since  $N$ is related to $B$ by $N =\frac{B(2B-1)(2B+1)}{3}$, the measurement number $m$ should depend on $N$ as $O(N^{1/3})$. We propose a similar preconditioning technique to improve this bound. The following inequality is crucial for our derivations:
\[
 \card{(\sin \theta)^{1/2} \mathrm d_l^{k,n}(\cos \theta)} \lesssim (2l +1)^{-1/4}.
\]
 We prove the above inequality in the appendix. This inequality suggests that the upper bound is improved if we precondition $\D lkn$ by $(\sin\theta)^{1/2}$. The preconditioning technique can be applied with Theorem \ref{thm:RIP_BOS} and Theorem \ref{thm:RIP_recovery} to yield the recovery guarantee for random sampling patterns, stated in the following theorem. 

 \begin{theorem}
 \label{thm:theorem_wigner}
Consider the problem \eqref{eq:LIP} of finding Fourier coefficients $\vec{g}$ of a band-limited function $g\in L^2(\SO(3))$ from noisy linear measurements $\mathbf y=\mathbf A\mathbf g+\boldsymbol{\eta}$ with $\norm{\vec{\eta}}_{\infty} \leq \epsilon$. 

Suppose that the sensing matrix $\mat{A}$ is constructed as \eqref{eq:sensingmatrix} using $m$ i.i.d. samples $(\theta_p,\phi_p,\chi_p)$, $p\in[m]$ drawn uniformly from $[0,\pi]\times[0,2\pi]\times[0,2\pi]$. 
 Let $\mathbf{P}$ be a diagonal matrix with {each diagonal element $P_{ii}= \sin(\theta_i)^{1/2}$ for $i \in [m]$}. The number of measurements $m$ is assumed to satisfy the following inequality 
\begin{equation*}
m \gtrsim \, {N}^{1/6} \,s\, \log^3(s) \,\log(N).
\end{equation*}
Then with  probability at least $1-N^{-\gamma log^3(s)}$, the following holds. If $\mathbf{g}^\#$ is the  solution to the following problem
\begin{equation*}
 \mathbf{g}^\# = \arg\min \left\Vert \mathbf{z} \right\Vert_1 \textnormal{subject to} \left\Vert \mathbf{P}\mathbf{A} \mathbf{z}- \mathbf{P}\mathbf{y}\right\Vert_2 \leq \sqrt{m}\epsilon. 
\end{equation*}
then, 
\begin{equation*}
 \left\Vert \mathbf {g} -\mathbf{g}^{\#}\right\Vert_2 \lesssim \frac{ \sigma_s (\mathbf {g})_1}{\sqrt{s}} + \epsilon.
\end{equation*}
In particular, when the measurements are not noisy, the recovery is unique for $s$-sparse signals, namely $\mathbf {g}=\mathbf{g}^{\#}$.
\end{theorem}
\begin{proof}
  The proof is given in Appendix \ref{proof:thm:theorem_wigner}.
\end{proof}

\begin{remark}
{
The recovery guarantee proofs for many algorithms are based on \ac{RIP}. Therefore our \ac{RIP} proof implies recovery guarantee for algorithms like iterative hard thresholding, hard thresholding pursuit and  orthogonal matching pursuit (see \cite[Remark 12.35]{foucart2013mathematical}).
}
\end{remark}

\begin{remark}
{
The role of preconditioning matrix is to counter the increase of Wigner D-functions at the endpoints of the interval. As we discussed above, there is a more general result based on pre-conditioning given in \cite{burq2012weighted}. Their results applies to the functions that are canonical solutions to Laplacian defined over a compact $n$-dimensional Riemannian manifold. Spherical harmonics and Wigner D-functions belong to this class of functions. It has been shown that the first $N$ canonical solutions, called eigenfunctions, defined on a compact $n$-dimensional Riemannian manifold are uniformly bounded by $N^{n-1/2n}$ \cite[Corollary 2]{burq2012weighted}. 
For $\SO(3)$, a 3-dimensional compact manifold, this approach yields the bound $N^{1/3}$ which is worse than the result above. As stated in \cite{burq2012weighted}, this bound deteriorates as the dimension of underlying manifold increases. There is another more powerful  result  in \cite{rauhut2012sparse,burq2012weighted} with better scaling with $N$. This result applies to functions defined over surfaces of revolution. However, Wigner D-functions are not defined for surfaces of revolution, and therefore these results do not apply. In the numerical results, we also consider the performance of preconditioning and measure in \cite{burq2012weighted}. It is, however, not clear at the moment how a similar bound can be obtained for Wigner D-functions.
}
\end{remark}

\section{Coherence Analysis of Sensing Matrices for Regular Sampling Patterns} \label{Sect:coherence}
Theorem \ref{thm:theorem_wigner} guarantees that random samples are suitable for sparse recovery of sparse Wigner D-expansion, while a similar result for spherical harmonics was given in \cite{rauhut2011sparse}. Practitioners use, however, more deterministic and regular samples. For instance, the samples in antenna design applications  are taken through robotic probes, which have physical limitations for taking too close measurements.  Therefore sampling patterns that are sufficiently distant and lead to smoother probe movements are preferred. In practice, the sampling points  are chosen from some known structures like equiangular sampling patterns. The main challenge is to find suitable regular patterns for sparse recovery.

Verifying \ac{RIP} for deterministic sensing matrices is computationally hard. Furthermore, except the single example of \cite{bourgain_explicit_2011}, only randomly generated sensing matrices have been shown so far to satisfy \ac{RIP}. 
There are, however, examples of matrices that do not satisfy \ac{RIP} and yet provide provable recovery guarantees \cite{dirksen_gap_2016}. 
Therefore, instead of using \ac{RIP}, we choose another notion to assess whether a sensing matrix is suitable for solving inverse problems. There are other concepts for evaluating the \textit{goodness} of sensing matrices, such as spark or mutual coherence of a matrix. {The mutual coherence has been used  to construct deterministic sensing matrices. For Fourier basis, the authors in  \cite{xia2005achieving} used tools from combinatorial number theory, in this case difference sets, to construct deterministic partial Fourier matrices for specific choice of input dimension $N$ and measurement numbers $m$. This construction was shown to achieve the Welch bound. When the input dimension $N$ is prime and  with specific $m \leq N$, the authors in \cite{haupt2010restricted} developed a method to produce deterministic Fourier matrices that can recover sparse signals with dimension $s \leq \frac{\sqrt{\epsilon}}{32} \left(\frac{\sqrt{\epsilon \log 2}}{2}\right)^{\exp{\left (\frac{4}{\epsilon} \right)}}
m^{\frac{1-\epsilon}{2}}$ for $\epsilon \in (0,1)$.
The authors in \cite{xu2014compressed} proposed a construction that can recover sparse signal with sparsity dimension $s < \frac{\sqrt{m}}{2(N-1)} + 0.5$ by using \ac{BP}.
} 

{In contrast to Fourier sensing matrices, there are only limited works related to the construction of deterministic sensing matrices from spherical harmonics and Wigner D-functions for compressed sensing. For instance, spiral sampling points are used to construct such sensing matrices, as investigated in \cite{alem2012sparse,cornelius2015analysis,hofmann2019minimum}, and perform numerical comparison of success recovery several compressed sensing algorithms. To the best of our knowledge, this paper is the first work to discuss the coherence bounds for those matrices.}
\begin{definition}
The mutual coherence of a matrix $\mbf A=[\vec{a}_1 \dots \vec{a}_N]\in\mbb C^{m\times N}$ is defined as the maximum of the normalized inner product of columns of the matrix, i.e.,  
\[
 \mu
 (\mbf{A})\defeq\max_{1\leq i<j\leq N}\frac{\card{\langle\vec{a}_i,\vec{a}_j\rangle}}{\norm{\vec{a}_i}_2  \norm{\vec{a}_j}_2
 }.
\]
\end{definition}
The mutual coherence belongs to the interval $[0,1]$. As a rule of thumb, the coherence of the sensing matrix should be very small for recovery of moderately sparse vectors. It is possible to obtain recovery guarantees for deterministic sensing matrices using its coherence value (for example see  \cite[Theorem 5.7]{foucart2013mathematical}). These results, however, yield bounds on the number of measurements that scale quadratically with the sparsity level. This is underwhelming even for moderate sparsity regime. Nevertheless, the coherence can still be used as a good indication for fitness of a sensing matrix, which is the approach we opt in this article. 

The mutual coherence expression for spherical harmonics, $\mu_1(\mbf A)$, and Wigner D-functions, $\mu_2(\mbf A)$, are given by
\begin{equation} \label{coherence_SH}
\begin{aligned}
\mu_1(\mat A)\defeq& &  \underset{ 1 \leq r < q \leq N}{\text{max}}
 \,\,\card{\sum_{p=1}^{m}\frac{ \mathrm{Y}_{l{(q)}}^{k{(q)}}(\theta_p,\phi_p)  \overline{\mathrm{Y}_{l{(r)}}^{k{(r)}}(\theta_p,\phi_p)}}{\norm{\mathrm{Y}_{l{(q)}}^{k{(q)}}(\boldsymbol\theta,\boldsymbol\phi)}_2 \norm{\mathrm{Y}_{l{(r)}}^{k{(r)}}(\boldsymbol\theta,\boldsymbol\phi)}_2}}
\end{aligned}
\end{equation}

\begin{equation} \label{coherence_Wigner}
\begin{aligned}
&\mu_2(\mat{A})\defeq  \underset{ 1 \leq r < q \leq N}{\text{max}}
 \,\,\card{\sum_{p=1}^{m}\frac{ \mathrm{D}_{l{(q)}}^{k{(q)},n{(q)}}(\theta_p,\phi_p,\chi_p)  \overline{\mathrm{D}_{l{(r)}}^{k{(r)},n{(r)}}(\theta_p,\phi_p,\chi_p)}}{\norm{\mathrm{D}_{l{(q)}}^{k{(q)},n{(q)}}(\boldsymbol\theta,\boldsymbol\phi,\boldsymbol\chi)}_2 \norm{\mathrm{D}_{l{(r)}}^{k{(r)},n{(r)}}(\boldsymbol\theta,\boldsymbol\phi,\boldsymbol\chi)}_2}},  
\end{aligned}
\end{equation}
where we adopt the following convention:
\[\Y{l}{k}(\vec\theta,\vec\phi)\defeq
\begin{pmatrix}
 \mathrm{Y}_{l}^{k}(\theta_1,\phi_1)\\
 \vdots\\
 \mathrm{Y}_{l}^{k}(\theta_m,\phi_m). 
\end{pmatrix}
\]
and 
\[\mathrm{D}_{l}^{k,n}(\boldsymbol\theta,\boldsymbol\phi,\boldsymbol\chi)\defeq
\begin{pmatrix}
 \mathrm{D}_{l}^{k,n}(\theta_1,\phi_1,\chi_1)\\
 \vdots\\
  \mathrm{D}_{l}^{k,n}(\theta_m,\phi_m,\chi_m)
\end{pmatrix}.
\]

As a reminder, the problem of designing sensing matrix for spherical harmonics and Wigner D-expansion boils down to finding the sequence of azimuth, elevation, and for Wigner D-functions case, polarization over which the measurements are taken. For spherical harmonics, the sampling pattern is given by pairs $(\theta_p,\phi_p)$ with $p\in[m], \theta_p \in [0,\pi]$ and $\phi_p \in [0,2\pi)$. For Wigner D-expansion, a rotation variable should be added and the sampling pattern is given by pairs $(\theta_p,\phi_p,\chi_p)$ with $p\in[m], \theta_p \in [0,\pi]$ and $\phi_p,\chi_p \in [0,2\pi)$. In the next section, our first result states that many sampling patterns, which are widely used in practice, have high mutual coherence and therefore are {inapplicable} for compressed sensing.

\subsection{Modularly Symmetric Patterns over Azimuth and Polarization}

A large class of regular sampling patterns select their sampling patterns on a regular grid over $\theta,\phi$ and $\chi$. 
Some of these sampling patterns, however, would lead to high mutual coherence and therefore should be avoided for compressed sensing applications. Spherical harmonics and Wigner D-functions are defined by associated Legendre polynomials and Jacobi polynomials. These polynomials are linearly related to each other for different orders and degrees. Through this relation, two columns of the sensing matrix can become strongly coherent in some cases.  The following theorem concerns one of these cases. It states the regular sampling on $\phi$ and $\chi$ can lead to full coherence. 
 
\begin{theorem}\label{thm:fullcoherence}
Let the matrix $\mat{A} \in \mathbb{C}^{m \times N}$ be constructed from samples of spherical harmonics $\Y{l}{k}(\theta,\phi)$ or Wigner D-functions $\D{l}{k}{n}(\theta,\phi,\chi)$. {For a signal with bandwidth $B$, suppose that a given sampling pattern for orders $-(B-1)\leq k,n\leq B-1$ satisfies:}
\begin{align}
2k\phi_i&\equiv 2k\phi_j \mod 2\pi,&\forall i,j\in[m]\label{eq:mod_angle_SH}\\
2n\chi_i+2k\phi_i&\equiv 2n\chi_j +2k\phi_j\mod 2\pi,&\forall i,j\in[m]\label{eq:mod_angle_WD}
\end{align}
{respectively for spherical harmonics and Wigner D-functions.} Then the mutual coherence of this matrix attains its maximum, i.e., $\mu(\mat{A}) = 1$.
\end{theorem}
\begin{proof}
Associated Legendre polynomials satisfy a symmetry relation over order in the following sense  \cite{lohofer1998inequalities}:
 \begin{equation}
P_l^{-k}(\cos \theta) = (-1)^k C_{lk} P_l^k(\cos \theta)
\end{equation}
where {$C_{lk} = \frac{(l-k)!}{(l+k)!}$}. This relation implies immediately a symmetric relation over orders of spherical harmonics, namely
 \begin{equation}
  \Y l{-k}(\theta,\phi)=(-1)^k\overline{\Y l{k}(\theta,\phi)}=(-1)^k{\Y l{k}(\theta,\phi)}e^{-\i 2k\phi}.
  \label{eq:symmerty-SH}
 \end{equation}
 Now if the azimuth sampling points are selected as $2k\phi_i\equiv 2k\phi_j \mod 2\pi$ for all $i,j\in[m]$, then the equality $e^{-\i 2k\phi_i}=e^{-\i 2k\phi_j}$ holds, which implies:
\[
\Y l{-k}(\boldsymbol\theta,\boldsymbol\phi)=C_k\Y l{k}(\boldsymbol\theta,\bld\phi) 
\]

for some constant $C_k$. This means that there are two columns of the matrix, corresponding to these two basis functions, totally coherent with each other and therefore yielding the coherence equal to one. On the other hand, it can be easily seen that by inverting the sign of orders of Wigner D-functions, the orders of respective Jacobi polynomial does not change and therefore:
\begin{equation}
\mathrm{d}_{l}^{k,n}(\cos \theta) =(-1)^{n-k}\mathrm{d}_{l}^{-k,-n}(\cos \theta).
\label{eq:symmetry-Wd}
\end{equation}
which means that 
\begin{equation*}
\begin{aligned}
\D{l}{k}{n}(\theta,\phi,\chi)& = (-1)^{n-k}\overline{\D{l}{-k}{-n}(\theta,\phi,\chi)}\\ &=(-1)^{n-k}\D{l}{-k}{-n}(\theta,\phi,\chi)e^{-j2k\phi}e^{-j2n\chi}.
\end{aligned}
\end{equation*}
If for some $k,n$, we have $2n\chi_i+2k\phi_i\equiv 2n\chi_j +2k\phi_j\mod 2\pi$ for all $i,j\in[m]$, then similar to spherical harmonics, it holds that:
\[
\D{l}{k}{n}(\bld\theta,\bld\phi,\bld\chi)=(-1)^{n-k}\D{l}{-k}{-n}(\bld\theta,\bld\phi,\bld\chi).
\]

And therefore there are two columns that are completely coherent and therefore the mutual coherence is equal to one.
\end{proof}
The previous theorem precludes some of familiar sampling patterns. One notable example is equiangular sampling on $\phi$ namely,  $\phi_p = \frac{2\pi(p-1)}{m-1}$ for $p \in [m]$. If the number of samples are odd and smaller than $2B-1$, the sensing matrix has the coherence equal to one with columns corresponding to $k=\frac{m-1}2$ being completely coherent. 
For Wigner D-functions, the equiangular samples on the azimuth $\phi$ and polarization $\chi$ are not proper sampling patterns. Note that in Wigner D-functions case, it is possible to end up with full coherence even if the polarization and azimuth angles are chosen irregularly. 

Theorem \ref{thm:fullcoherence} provides a first step to understand what to avoid  in sensing matrix designs. In the next sections, we first provide an alternative way of characterizing coherence using tools originally developed in quantum mechanics. Afterwards, instead of imposing regularity on $\phi$ and $\chi$, we study regular sampling on the elevation $\theta$. 
\subsection{Coherence Analysis using Wigner 3j Symbols}
\label{sec:clebschgordancoherence}
Spherical harmonics and Wigner D-functions express wave functions in the study of angular momentum in quantum mechanics. Their products appear in the characterization of total angular momenta of a composite system in terms of the angular momentum of its two sub-systems. This characterization involves a decomposition of the wave function into two wave functions with different angular momenta. The coefficients of this decomposition are given by the Clebsch-Gordan coefficients, also known as Wigner or vector coupling coefficients, as well as Wigner 3j symbols \cite{griffiths_introduction_2014,rose1995elementary,edmonds_angular_2016,wigner2012group,biedenharn_angular_1981}. We focus on the latter and provide briefly some of the useful identities here.
Wigner 3j symbols are  denoted by {$ \begin{pmatrix}
   l_1 & l_2 & l_3 \\
   k_1 & k_2 & k_3
  \end{pmatrix} \in \R$}, and their exact formula is given in \cite[Section 7.10.2]{kennedy_hilbert_2013} or \cite{edmonds_angular_2016}. {In quantum mechanics, $l_i$'s and $k_i$'s are non-negative integers or half-odd numbers, however in this paper, we only focus on the case where they are all integers.} Despite their complex expressions, Wigner 3j symbols have a few useful properties. The so-called \textit{selection rules} state that Wigner 3j symbols $\begin{pmatrix}
   l_1 & l_2 & {l}_3 \\
   k_1 & k_3 & k_3
  \end{pmatrix}$
are non-zero only if:
\begin{itemize}
\item The absolute value of $k_i$ does not exceed $l_i$, i.e., $-l_i \leq k_i \leq l_i$ for $i=1,2,3$
\item The summation of all $k_i$ should be zero: $k_1 + k_2 +k_3 = 0$.
\item Triangle inequality holds for $l_i$'s: $\card{l_1 - l_2} \leq l_3 \leq l_1 + l_2$.
\item The sum of all $l_i$'s should be an integer.
\item If $k_1 = k_2 = k_3 = 0$,  $l_1+l_2+l_3$ should be an even integer.
\end{itemize}
If one of the above conditions does not hold, the corresponding Wigner 3j symbol will be zero.
In coherence analysis of the sensing matrix in  \eqref{coherence_Wigner} and \eqref{coherence_SH}, one encounters sums over products of spherical harmonics or Wigner D-functions. We can use Wigner 3j symbols to express {these sums in terms of sums of spherical harmonics}, or respectively Wigner D-functions. The decomposition reveals in another way the effect of sampling patterns on the mutual coherence. The following proposition, derived from the decomposition based on Wigner 3j symbols, characterizes the inner product between two columns of the sensing matrix.
\begin{proposition} \label{prop:CG_innerproduct}
Let $\mathrm{D}_{l }^{k ,n }(\theta ,\phi ,\chi )$ be the Wigner D-function with degree $l$ and orders $k,n$, and let $\Y{l}{k}(\theta,\phi)$ be the spherical harmonics with degree $l$ and order $k$. Then the following identities hold:
\begin{equation}
\small
\begin{aligned}
 &\sum_{p=1}^m\overline{\mathrm{D}_{l_1}^{k_1,n_1}(\theta_p,\phi_p,\chi_p)}  \mathrm{D}_{l_2}^{k_2,n_2}(\theta_p,\phi_p,\chi_p)  \\
 &= C_{k_2,n_2}\sum_{\substack{\hat l=|l_2-l_1| }}^{l_1 + l_2} {\sqrt{\frac{(2l_1 + 1)(2l_2 +1)(2\hat{l} + 1)}{8 \pi^2}}} \\& \times \begin{pmatrix}
   l_1 & l_2 & \hat{l} \\
   -n_1 & n_2 & -\hat{n} 
  \end{pmatrix} \begin{pmatrix}
   l_1 & l_2 & \hat{l} \\
   -k_1 & k_2 & -\hat{k}
  \end{pmatrix}  \paran{\sum_{p=1}^m\mathrm{D}_{\hat{l}}^{\hat{k},\hat{n}}(\theta_p,\phi_p,\chi_p)},
\end{aligned}
\end{equation}
\begin{equation}
\small
\begin{aligned}
  &\sum_{p=1}^m\overline{\mathrm{Y}_{l_1}^{k_1}(\theta_p,\phi_p)}  \mathrm{Y}_{l_2 }^{k_2}(\theta_p,\phi_p)=(-1)^{k_1}\mathrm{Y}_{l_1}^{-k_1}(\theta_p,\phi_p)  \mathrm{Y}_{l_2}^{k_2}(\theta_p,\phi_p)\\
  &=(-1)^{k_2}\sum_{\hat{l}=\card{l_1 - l_2} }^{l_1+l_2}\sqrt{\frac{(2l_1+1)(2l_2+1)(2\hat l+1)}{4\pi}} \\& \times \begin{pmatrix}
   l_1 & l_2 & \hat{l} \\
   0  & 0  & 0 
  \end{pmatrix} \begin{pmatrix}
   l_1 & l_2 & \hat{l} \\
   -k_1  & k_2  & -\hat{k}
  \end{pmatrix}\paran{\sum_{p=1}^m\mathrm{Y}_{\hat{l}}^{\hat{k}}(\theta_p,\phi_p)}.
\end{aligned}
\end{equation}
where $\hat{k} = k_2-k_1$ and $\hat{n} = n_2 - n_1$ and the phase factor $ C_{k_2,n_2} = (-1)^{k_2 + n_2}$.
\end{proposition}
\begin{proof}
The product of two Wigner D-functions of degrees $l_1$ and $ l_2$   and orders  ${k_1},{n_1}$ and ${k_2},{n_2}$ writes in terms of 
the Wigner 3j symbols as

\begin{equation}
\begin{aligned}
&\mathrm{D}_{l_1}^{k_1,n_1}(\theta,\phi,\chi) \mathrm{D}_{l_2}^{k_2,n_2} (\theta,\phi,\chi) = \\
&(-1)^{\hat{k}+\hat{n}}  \sum_{\hat{l}=\card{l_1 - l_2}}^{l_1+l_2} {\sqrt{\frac{(2l_1 + 1)(2l_2 +1)(2\hat{l} + 1)}{8 \pi^2}}}
\begin{pmatrix}
   l_1 & l_2 & \hat{l} \\
   k_1 & k_2 & -\hat{k}
  \end{pmatrix}\times \begin{pmatrix} 
   l_1 & l_2 & \hat{l} \\
   n_1 & n_2 & -\hat{n}
  \end{pmatrix}  \mathrm{D}_{\hat{l}}^{\hat{k},\hat{n}}(\theta,\phi,\chi),
\end{aligned}
\label{eq:3j-Wigner}
\end{equation}
where $\hat n=n_1+n_2$ and $\hat k=k_1+k_2$  \cite[pp. 61-62]{edmonds_angular_2016}. The spherical harmonics version of the expansion can be obtained by using $n_1=n_2=0$.  

From the conjugate property of these functions, we know that:

\[
\overline{\mathrm{D}_{l_1}^{k_1,n_1}(\theta,\phi,\chi)}  =  (-1)^{k_1-n_1}\D{l_1}{-k_1}{-n_1}(\theta,\phi,\chi)
\text{ and }\overline{\Y{l_1}{k_1}(\theta,\phi)}=(-1)^{k_1}\Y{l_1}{-k_1}(\theta,\phi).
\] 
The proof follows with standard manipulations by plugging in these identities to \eqref{eq:3j-Wigner}.
\end{proof}

According to Proposition \ref{prop:CG_innerproduct}, the inner product between columns of the sensing matrix depends on the sampling pattern through the sum $\sum_{p=1}^m\mathrm{Y}_{\hat{l}}^{\hat{k}}(\theta_p,\phi_p)$ or ${\sum_{p=1}^m\mathrm{D}_{\hat{l}}^{\hat{k},\hat{n}}(\theta_p,\phi_p,\chi_p)}$. The next theorem uses this characterization when the elevation samples are chosen symmetrically in the following sense. 

\begin{definition}[Cosine-symmetric sampling]
  Cosine-symmetric sampling patterns are defined by a set of $m$ samples $(\theta_p,\phi_p,\chi_p)$ for $p=1,\dots,m$  such that the set $\{\cos\theta_1,\dots,\cos\theta_m\}$ consists of symmetric points around the origin inside $[-1,1]$. 
\end{definition}

\begin{theorem}
\label{thm:cosinesym_orth}
  Suppose that $m$ samples are chosen such that the elevation samples $\theta_1,\dots,\theta_m$ are cosine-symmetric. Consider two columns of the sensing matrix corresponding to samples of two spherical harmonics with equal order $k_1=k_2$ and different degrees $l_1$ and $l_2$. If $l_1+l_2$ is odd, then the columns are orthogonal.  The same conclusion holds for two Wigner D-functions when one pair of orders are equal and the other pair of orders are equal to zero. 
\end{theorem}

\begin{proof}
We start with spherical harmonics. We use Proposition \ref{prop:CG_innerproduct}. Note that:
\begin{align*}
\D l00(\theta_p,\phi_p,\chi_p)&=\frac{1}{2\pi}\Y l0(\theta_p,\phi_p)=\sqrt{\frac{2l+1}{8\pi^2}} P^0_l(\cos\theta_p)\\
&=\sqrt{\frac{2l+1}{8\pi^2}} P_l(\cos\theta_p),
\end{align*}
where $P_l(\cos\theta)$ is the Legendre polynomial. Legendre polynomials are odd functions for odd $l$. This means that for the cosine-symmetric elevation sampling, when $l$ is odd, it holds that:
\[
\sum_{p=1}^m P_l(\cos\theta_p)=0.
\]
Therefore Proposition \ref{prop:CG_innerproduct} implies that:
\begin{equation}
\small
\begin{aligned}
  &\sum_{p=1}^m\overline{\mathrm{Y}_{l_1}^{k}(\theta_p,\phi_p)}  \mathrm{Y}_{l_2 }^{k}(\theta_p,\phi_p)=\\
  &(-1)^{k}\sum_{\hat{l}=\card{l_1 - l_2},\text{even} }^{l_1+l_2}\sqrt{\frac{(2l_1+1)(2l_2+1)(2\hat l+1)}{4 \pi}}  \times \begin{pmatrix}
   l_1 & l_2 & \hat{l} \\
   0  & 0  & 0 
  \end{pmatrix} \begin{pmatrix}
   l_1 & l_2 & \hat{l} \\
   -k  & k  & 0
  \end{pmatrix}\paran{\sum_{p=1}^m\mathrm{Y}_{\hat{l}}^{0}(\theta_p,\phi_p)}.
\end{aligned}
\end{equation}
On the other hand, according to the selection rules, if $l_1+l_2$ is odd and $\hat{l}$ is even, then $\begin{pmatrix}
   l_1 & l_2 & \hat{l} \\
   0  & 0  & 0 
  \end{pmatrix}=0$, which proves the theorem. A similar argument works for the Wigner D-functions.
\end{proof}

Theorem \ref{thm:cosinesym_orth} implies that, if the elevation sampling pattern is cosine-symmetric, there are at least $\lfloor\frac{B}{2}\rfloor$ columns that are mutually orthogonal. Cosine-symmetric sampling patterns are also regular, hence, suitable for practical measurements. Using this insight, in the next section, we propose a  cosine-symmetric pattern with minimal coherence.


\section{Equispaced Elevation Sampling for Spherical Harmonics and Wigner D-Functions}
\label{Sect:equispacedSampling}

As we discussed, among regular sampling patterns, equiangular sampling patterns on azimuth and polarization lead to coherent, and therefore undesirable, sensing matrices. On the other hand, a class of regular sampling patterns on the elevation yield incoherent measurements as in Theorem \ref{thm:cosinesym_orth}. 

As soon as the elevation sampling is fixed, the mutual coherence is automatically bounded from below regardless of the choice of azimuth sampling patterns. This is because, in the inner products of columns with equal orders $k_1 = k_2 = k$ and $n_1=n_2=n$, the terms $e^{\i k_1\phi_p}$ and $e^{-\i k_2\phi_p}$ and the terms $e^{\i n_1\chi_p}$ and $e^{-\i n_2\chi_p}$ cancel each other out. Furthermore the $\ell_2$-norm of $ \mathrm{Y}_{l}^{k}(\boldsymbol\theta,\boldsymbol\phi)$ and $\mathrm{D}_l^{k,n}(\bld\theta,\bld\phi,\bld\chi)$ depends only on elevation sampling for all degrees and orders. We state this simple result in the following proposition.

\begin{proposition}\label{prop:lowerbound}
  Let the elevation sampling be fixed to $\theta_1,\theta_2,\dots,\theta_m$. For all possible choices of azimuth $\phi_p$, and polarization $\chi_p$, $p\in[m]$, it holds that
  \begin{align*}
\mu_1(\mathbf A) & \geq \max_{\substack{{l\neq r}\\{|k|\leq \min{(l,r)}}}}\frac{\bigl|\sum_{p=1}^{m} P_{l}^{k}(\cos \theta_p)P_{r}^{k }(\cos \theta_p)\bigr|}{\norm{P_{l}^{k}(\cos\bld \theta)}_2 \norm{P_{r}^{k}(\cos \bld\theta)}_2},\\
\mu_2(\mathbf A) & 
\geq \max_{\substack{{l\neq r}\\{|k|,|n|\leq \min{(l,r)}}}}
\frac{\bigl|\sum_{p=1}^{m} \d_{l}^{k,n}(\cos \theta_p)\d_{r}^{k,n}(\cos \theta_p)\bigr|}{\norm{\d_{l}^{k,n}(\cos\bld \theta)}_2 \norm{\d_{r}^{k,n}(\cos \bld\theta)}_2},
\end{align*}
where 
\begin{align*}
  P_{l}^k (\cos \boldsymbol\theta)&\defeq\paran{{P}_{l}^k(\cos \theta_1),\dots,{P}_{l}^k(\cos \theta_m)}^T\\
  \d_{l}^{k,n}(\cos\bld \theta)&\defeq\paran{\d_l^{k,n}(\cos\theta_1),\dots,\d_l^{k,n}(\cos\theta_m)}^T.
\end{align*}
In particular it holds that
\begin{equation*}
\begin{aligned}
\min\bracket{\mu_1(\mat A),\mu_2(\mat A)} \geq \frac{\biggl|\sum_{p=1}^{m} {P}_{B-1}(\cos \theta_p), {P}_{B-3}(\cos \theta_p) \biggr|}{\norm{P_{B-1} (\cos \boldsymbol\theta)}_2 \norm{P_{B-3} (\cos \boldsymbol\theta)}_2}
\end{aligned}
\end{equation*}
where ${P}_l (\cos \theta)$ is the Legendre polynomial of degree $l$ and
\[
 P_{l} (\cos \boldsymbol\theta)\defeq \paran{{P}_{l}(\cos \theta_1),\dots,{P}_{l}(\cos \theta_m)}^T.
\]
\end{proposition}

The proposition follows by choosing equal orders in the definition of the coherence. Its lower bounds hold in general for any sampling pattern. Note that Theorem \ref{thm:cosinesym_orth} implies that:
\[
\sum_{p=1}^{m} {P}_{B-1}(\cos \theta_p) {P}_{B-2}(\cos \theta_p)=0. 
\]
This is why the lower bound involves only Legendre polynomials of degree $B-1$ and $B-3$.  

On the face of it, Proposition \ref{prop:lowerbound} seems trivial. 
It indicates the sensitivity of mutual coherence to the choice of elevation sampling alone. The lower bound, however, is almost tight for a class of regular sampling patterns on elevation defined below if $m$ is sufficiently large. 

\begin{definition}[Equispaced Elevation Sampling]
\label{def:equispaced}
The equispaced elevation sampling pattern is defined by the elevation samples $\theta_p$ for $p\in[m]$ given by
\[
  \cos \theta_{p} = \frac{2p-m-1}{m-1},
\]
which satisfies \(-1=\cos \theta_{1}  <\cos \theta_2<\hdots< \cos \theta_{m-2} < \cos \theta_{m-1} < \cos \theta_{m}=1 \).
\end{definition}
{
Note that the above sampling points are cosine-symmetric.} 
For the equispaced elevation sampling, for sufficiently large $m$, the dominant inner product among all the inner products between the spherical harmonics of equal orders is the inner product between degrees of $B-1$ and $B-3$. This can be clearly seen in Fig. \ref{fig:SHdegrees}. After a certain measurement number $m$, the inner products between columns of equal orders are completely ordered. The ordering of inner products between two columns, say of degree $l_1$ and $l_2$, corresponds to a partial order defined on the degree pairs $(l_1,l_2)$. This can formally proven. We relegate, however, the detailed derivations of this result to another work \cite{bangun_toappear_2019}.  The lower bound is therefore tight in the following sense. Once the elevation sampling pattern is equispaced, there is a fundamental lower bound on the coherence independent of the choice of azimuth and polarization. This lower bound is given in Proposition \ref{prop:lowerbound} for sufficiently large $m$.  Note that the number of measurements $m$ should be of $O(N^{1/2})$ for the tightness of the lower bound in this sense. This dependence on $N$ is in general undesirable and cannot be removed, as it can be seen in the numerical result. The exact inequality, however, involves large constants, so that, for many $N$'s of practical interest, the number of required measurements for the tightness of the lower bound are small. For example, when $N=1024$, Figure \ref{fig:SHdegrees} show that after 100 measurements, the lower bound becomes tight. In the next section, we provide a way to choose azimuth sampling patterns that achieves the lower bound for spherical harmonics. 
 
\begin{figure}[htb!]
\begin{centering}  
     \scalebox{0.6}{\input{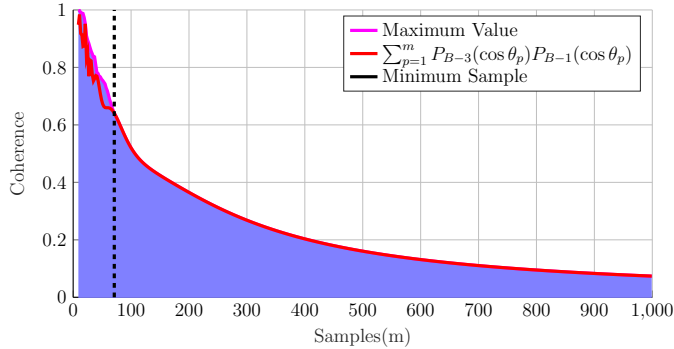}}  
  
\caption{{The inner product of two columns of the sensing matrix for different measurement numbers $m$ for $B=32$ ($N=1024$)}}  
\label{fig:SHdegrees}
\end{centering}  
\end{figure} 

\subsection{Sampling Pattern Design using Coherence Minimization}

Is the lower bound of Proposition \ref{prop:lowerbound} tight for equispaced sampling patterns? That is, can we find an azimuth sampling pattern that achieves the bound? To do so, we directly minimize the mutual coherence as a figure of merit. The problem of minimizing the mutual coherence for spherical harmonics and Wigner D-functions is non-convex in general since Legendre polynomials, Jacobi polynomials and trigonometric polynomials $e^{\i(k^{(r)}-k{(q)})\phi_p}$ are non-convex. We provide, however, a pattern search algorithm for minimizing the mutual coherence \cite{bangun2018}. Pattern search, however, requires less computation time and provides better results in comparison. It is particularly useful as it does not need to calculate the gradient during optimization process. There is, however, no guarantee that the method will converge to the global optimum. See for example  \cite{torczon1997convergence} for a discussion on the convergence of this algorithm. Although the method has rooms for improvements, it still yields, as we will see, sufficiently good sampling patterns.

First consider spherical harmonics.  The algorithm is described in  Algorithm \ref{algo_ps}. It starts by choosing initial $\vec{\phi}_0$
drawn uniformly at random on the interval $[0,2\pi)^m$. The elevation sampling pattern $\vec{\theta}$ is fixed. 

The algorithm has two hyperparameters $\lambda$ and $\Delta_0$. The parameter $\Delta_0$ is the initial update step, and determines the  search space, which is spanned along the canonical bases. The update step is decreased iteratively by the decay parameter $\lambda$. 
The algorithm tries to find the minimum coherence and its minimizer by checking the neighbor vectors where the initial update step is given as $\Delta_0$. The mutual coherence at the iteration $k$ is denoted by $\mu(\vec{\theta},\vec{\phi}_k)$. If the search fails,  the step size is decreased by scaling with $\lambda$. The algorithm stops when the number of iteration is achieved a pre-determined maximum or when the difference between the update coherence and the lower bound of Proposition \ref{prop:lowerbound}, denoted by $\mu_{\mrm{LB}}$, is small $\card{\mu(\vec{\theta},\vec{\phi}_k) - \mu_{\mrm{LB}}} \leq \epsilon$.

 \begin{algorithm}
    \caption{{Pattern search}}\label{algo_ps}
    \begin{algorithmic}
      \scriptsize
      \STATE \textbf{Initialization : } 
       \begin{itemize}
           \item $\boldsymbol \theta $ is given.
           \item $\boldsymbol{\phi}_0 \in \mathbb{R}^m$ as initial points.
           \item $\Delta_0 > 0$ as initial update step.
           \item Standard basis $\boldsymbol e_i$ for $i\in [m] $.
           \item Scaling for update rule $\lambda \in (0,1)$.
           \item Coherence of pair $\boldsymbol \theta, \boldsymbol \phi \in \R^m$ is given as $\mu(\boldsymbol \theta, \boldsymbol \phi)$. 
       \end{itemize}
    \FOR  {$k =0,\dots, k_{\max}$ until $\card{\mu(\vec{\theta},\vec{\phi}_k) - \mu_{\mrm{LB}}} \leq \epsilon$}
    \STATE{Create the set $S_k := \{\boldsymbol \phi_k \pm \Delta_k \boldsymbol e_i: i\in[m] \}$}
    \IF{ there is an $\mbf x \in S_k$ such that $\mu(\boldsymbol \theta,\mbf x) < \mu(\boldsymbol \theta,\boldsymbol \phi_k)$  } 
    \STATE{$\boldsymbol \phi_{k+1} = \mbf x\mod 2\pi$}\STATE{$\Delta_{k+1} = \Delta_{k}$}
    \ELSE
    \STATE{$\boldsymbol \phi_{k+1} = \boldsymbol \phi_{k} \mod 2\pi$}\STATE{$\Delta_{k+1} =  \lambda\Delta_{k}$}
    \ENDIF
      \ENDFOR   
    \end{algorithmic}
  \end{algorithm} 
Figure \ref{Coherence_comparison} compares the mutual coherence of the resulting sampling pattern from Algorithm \ref{algo_ps} with other sampling patterns widely used in applications. We use spiral \cite{saff1997distributing}, Hammersley \cite{cui1997equidistribution}, Fibonacci  \cite{swinbank2006fibonacci} and equiangular sampling patterns {to verify the result of Theorem \ref{thm:fullcoherence}}. {Another sampling pattern on the sphere is the so-called \textit{t-design} \cite{delsarte1977spherical}. Unfortunately, the spherical $t$-design does not exist for an arbitrary pair $m,t$ as given in \cite{womersley2018efficient}, which also restricts the flexibility to choose an arbitrary number of samples. To the best of our knowledge there is nothing related to spherical designs on the rotation group.}

The bandwidth of spherical harmonics is chosen as $B=10$, which yields  $N = B^2 = 100$. We plot also the Welch bound, which is the strict lower bound on the coherence of any $m\times N$ matrix. Figure \ref{Coherence_comparison}, interestingly, shows that the obtained sampling pattern achieves the lower bound of Proposition \ref{prop:lowerbound} and outperforms with a large margin the other sampling patterns. We have numerically observed that the lower bound can be achieved using our sampling patterns for $N$ {up to $10000$}. Figure \ref{Proposed} shows the distribution of this sampling points on the sphere for different number of samples $m$, $B=32$ and $N=B^2=1024$.

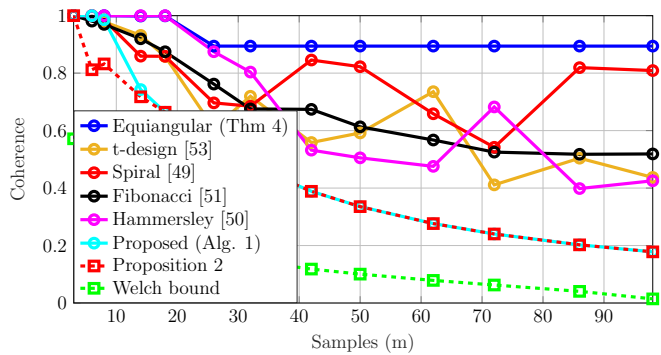
\begin{figure}[!htb]
  \begin{centering}
     \scalebox{0.6}{
%
%
\definecolor{mycolor1}{rgb}{0.92941,0.69412,0.12549}%
\definecolor{mycolor2}{rgb}{1.00000,0.00000,1.00000}%
\definecolor{mycolor3}{rgb}{0.00000,1.00000,1.00000}%
\begin{tikzpicture}

\begin{axis}[%
width=5in,
height=2.5in,
at={(3.328in,1.888in)},
scale only axis,
xmin=3,
xmax=98,
xlabel style={font=\color{white!15!black}},
xlabel={Samples (m)},
ymin=0,
ymax=1,
ylabel style={font=\color{white!15!black}},
ylabel={Coherence},
axis background/.style={fill=white},
title style={font=\bfseries},
xmajorgrids,
ymajorgrids,
legend style={at={(0,0)}, anchor=south west,legend cell align=left, align=left, draw=white!15!black}
]
\addplot [color=blue, line width=2.0pt, mark size=3.0pt, mark=o, mark options={solid, blue}]
  table[row sep=crcr]{%
3	1\\
6	1\\
8	1\\
14	0.998631050161952\\
18	0.999184666875425\\
26	0.894007333521563\\
32	0.894007333521561\\
42	0.894007333521561\\
50	0.894007333521561\\
62	0.894007333521561\\
72	0.894007333521561\\
86	0.894007333521561\\
98	0.894007333521561\\
};
\addlegendentry{Equiangular (Thm \ref{thm:fullcoherence})}

\addplot [color=mycolor1, line width=2.0pt, mark size=3.0pt, mark=o, mark options={solid, mycolor1}]
  table[row sep=crcr]{%
3	1\\
6	1\\
8	0.978653654079186\\
14	0.931427933743185\\
18	0.858062916086929\\
26	0.608478195961547\\
32	0.719461929152653\\
42	0.558546025289891\\
50	0.591898435317504\\
62	0.735463637609125\\
72	0.411119815162737\\
86	0.503997628173188\\
98	0.437004361277765\\
};
\addlegendentry{t-design \cite{womersley2018efficient}}

\addplot [color=red, line width=2.0pt, mark size=3.0pt, mark=o, mark options={solid, red}]
  table[row sep=crcr]{%
3	1\\
6	1\\
8	0.985534121909072\\
14	0.859031349054953\\
18	0.858526287374563\\
26	0.69645154125391\\
32	0.684804192780941\\
42	0.845582333109568\\
50	0.822365240002125\\
62	0.658558730475095\\
72	0.541916407026589\\
86	0.819146375726434\\
98	0.808943552600676\\
};
\addlegendentry{Spiral \cite{saff1997distributing}}

\addplot [color=black, line width=2.0pt, mark size=3.0pt, mark=o, mark options={solid, black}]
  table[row sep=crcr]{%
3	1\\
6	0.981023168630356\\
8	0.969594065126069\\
14	0.919533827424129\\
18	0.874369573823455\\
26	0.761851228315213\\
32	0.675347214239403\\
42	0.67389572298811\\
50	0.613174192256903\\
62	0.567237332038978\\
72	0.52509492173613\\
86	0.517567528659222\\
98	0.518749118539025\\
};
\addlegendentry{Fibonacci \cite{swinbank2006fibonacci}}

\addplot [color=mycolor2, line width=2.0pt, mark size=3.0pt, mark=o, mark options={solid, mycolor2}]
  table[row sep=crcr]{%
3	1\\
6	1\\
8	1\\
14	1\\
18	1\\
26	0.87446795829053\\
32	0.803695622199913\\
42	0.531997935245259\\
50	0.505177112362005\\
62	0.475522880230948\\
72	0.68230697421508\\
86	0.398309425414\\
98	0.425351182723522\\
};
\addlegendentry{Hammersley \cite{cui1997equidistribution}}

\addplot [color=mycolor3, line width=2.0pt, mark size=3.0pt, mark=o, mark options={solid, mycolor3}]
  table[row sep=crcr]{%
3	1\\
6	1\\
8	0.985534121909072\\
14	0.742828544416415\\
18	0.663863891252475\\
26	0.54937399248213\\
32	0.478699008657647\\
42	0.388675997849852\\
50	0.335406997322476\\
62	0.276422415351577\\
72	0.240244468550956\\
86	0.202427965002637\\
98	0.178054980916328\\
};
\addlegendentry{Proposed (Alg. \ref{algo_ps})}

\addplot [color=red, dashed, line width=2.0pt, mark size=3pt, mark=square, mark options={solid, red}]
  table[row sep=crcr]{%
3	1\\
6	0.811776877912782\\
8	0.832027338460355\\
14	0.717584828994402\\
18	0.663863891252475\\
26	0.54937399248213\\
32	0.478699008657647\\
42	0.388675997849852\\
50	0.335406997322476\\
62	0.276422415351576\\
72	0.240244468550955\\
86	0.202427965002637\\
98	0.178054980916328\\
};
\addlegendentry{Proposition \ref{prop:lowerbound}}

\addplot [color=green, dashed, line width=2.0pt, mark size=3pt, mark=square, mark options={solid, green}]
  table[row sep=crcr]{%
3	0.571488693325884\\
6	0.39780542762657\\
8	0.340824905430363\\
14	0.249096491442698\\
18	0.214512733147427\\
26	0.16955538549108\\
32	0.146508178831922\\
42	0.118105943749728\\
50	0.100503781525921\\
62	0.0786825372926715\\
72	0.0626751194241962\\
86	0.0405505916302066\\
98	0.0143576830751316\\
};
\addlegendentry{Welch bound}

\end{axis}
\end{tikzpicture}%
}  

   \caption{{The mutual coherence for different sampling patterns on sphere}}
    \label{Coherence_comparison}
    \end{centering}  
\end{figure}

\begin{figure}[!htb]
  \centering
     \scalebox{0.6}{\input{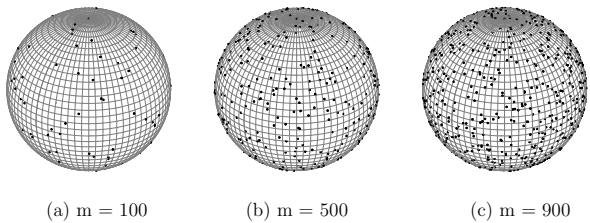}}  

   \caption{{Proposed sampling points}}
    \label{Proposed}
\end{figure}


 Algorithm \ref{algo_ps} can be extended to find pairs of $(\phi_p,\chi_p)$, $p \in [m]$, for Wigner D-functions. At each step, the algorithm searches simultaneously over the neighbor pairs, and advances similarly by updating $\Delta_k$ and $\vec{\phi}_k,\vec{\chi}_k$.
 The mutual coherence of the resulting sampling pattern is shown in Figure \ref{Coherence_comparison_Wigner_equiang} and is compared with other sampling patterns. The bandwidth is chosen as $B = 4$, hence, $N =\frac{B(2B-1)(2B+1)}{3} = 84$. It can be seen that the lower bound of Proposition \ref{prop:lowerbound} does not improve on the Welch bound for Wigner D-functions. Although the resulting sampling pattern outperforms significantly the other sampling patterns, it does not achieve the lower bound. This might be an artifact of our optimization method.

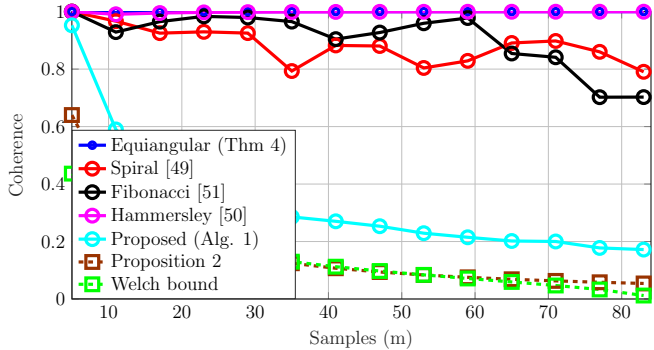
\begin{figure}[!htb]
  \centering
     \scalebox{0.6}{
%
%
\definecolor{mycolor1}{rgb}{1.00000,0.00000,1.00000}%
\definecolor{mycolor2}{rgb}{0.00000,1.00000,1.00000}%
\definecolor{mycolor3}{rgb}{0.60000,0.20000,0.00000}%
\begin{tikzpicture}

\begin{axis}[%
width= 5in,
height=2.5in,
at={(3.346in,1.735in)},
scale only axis,
xmin=5,
xmax=84,
xlabel style={font=\color{white!15!black}},
xlabel={Samples (m)},
ymin=0,
ymax=1,
ylabel style={font=\color{white!15!black}},
ylabel={Coherence},
axis background/.style={fill=white},
xmajorgrids,
ymajorgrids,
legend style={at={(0,0)}, anchor=south west, legend cell align=left, align=left, draw=white!15!black}
]
\addplot [color=blue, line width=2.0pt, mark size= 2.0pt, mark=o, mark options={solid, blue}]
  table[row sep=crcr]{%
5	1\\
11	1\\
17	1\\
23	1\\
29	1\\
35	1\\
41	1\\
47	1\\
53	1\\
59	1\\
65	1\\
71	1\\
77	1\\
83	1\\
};
\addlegendentry{Equiangular (Thm \ref{thm:fullcoherence})}

\addplot [color=red, line width=2.0pt, mark size= 4.0pt, mark=o, mark options={solid, red}]
  table[row sep=crcr]{%
5	1\\
11	0.966855165375045\\
17	0.925171690866185\\
23	0.92968235940336\\
29	0.92530758567095\\
35	0.793288606427342\\
41	0.882378093106189\\
47	0.880417150825572\\
53	0.804004481564544\\
59	0.828686139684846\\
65	0.890971268011396\\
71	0.897947992810131\\
77	0.860103816786042\\
83	0.790677966062383\\
};
\addlegendentry{Spiral \cite{saff1997distributing}}

\addplot [color=black, line width=2.0pt, mark size= 4.0pt, mark=o, mark options={solid, black}]
  table[row sep=crcr]{%
5	1\\
11	0.928749297293717\\
17	0.96514716348224\\
23	0.983517646521442\\
29	0.979498708852171\\
35	0.965254938563074\\
41	0.904300522537875\\
47	0.92726234110442\\
53	0.959463947620064\\
59	0.977755370910125\\
65	0.854323350851535\\
71	0.841142764530443\\
77	0.702281487848925\\
83	0.702599597192455\\
};
\addlegendentry{Fibonacci \cite{swinbank2006fibonacci}}

\addplot [color=mycolor1, line width=2.0pt, mark size= 4.0pt, mark=o, mark options={solid, mycolor1}]
  table[row sep=crcr]{%
5	1\\
11	0.990693075610521\\
17	0.995113783156175\\
23	0.997346038535144\\
29	0.998201558058819\\
35	0.998837466597405\\
41	0.999061195631201\\
47	0.999303756591646\\
53	0.999433877438808\\
59	0.999544288920045\\
65	0.999627587468256\\
71	0.999697085172426\\
77	0.999730836607927\\
83	0.999770892385236\\
};
\addlegendentry{Hammersley \cite{cui1997equidistribution}}

\addplot [color=mycolor2, line width=2.0pt, mark size=4pt, mark=o, mark options={solid, mycolor2}]
  table[row sep=crcr]{%
5	0.952426911204302\\
11	0.589978248675138\\
17	0.440203381072775\\
23	0.363709277326332\\
29	0.313482412812023\\
35	0.285682561017074\\
41	0.270335924611692\\
47	0.253282838320235\\
53	0.228980187093097\\
59	0.214743451622706\\
65	0.201850538006028\\
71	0.200132687025914\\
77	0.177729617524156\\
83	0.172002903115413\\
};
\addlegendentry{Proposed (Alg. \ref{algo_ps})}

\addplot [color=mycolor3, dashed, line width=2.0pt, mark size=4pt, mark=square, mark options={solid, mycolor3}]
  table[row sep=crcr]{%
5	0.64018439966448\\
11	0.35\\
17	0.241384193539774\\
23	0.183914120661846\\
29	0.148436219218273\\
35	0.124387786150071\\
41	0.107024722318206\\
47	0.0939049179872156\\
53	0.0836448225552256\\
59	0.0754026327437692\\
65	0.0686370499073718\\
71	0.0629842982467943\\
77	0.0581909239467796\\
83	0.0540749513129516\\
};
\addlegendentry{Proposition \ref{prop:lowerbound}}

\addplot [color=green, dashed, line width=2.0pt, mark size=4pt, mark=square, mark options={solid, green}]
  table[row sep=crcr]{%
5	0.436304304107961\\
11	0.282765253151662\\
17	0.21790836115787\\
23	0.178756545521525\\
29	0.151162233288431\\
35	0.129874823886379\\
41	0.112409561388812\\
47	0.0973896683464979\\
53	0.0839467459114694\\
59	0.071450447363369\\
65	0.0593445821064896\\
71	0.0469681501349579\\
77	0.0330951696160748\\
83	0.0120481927710843\\
};
\addlegendentry{Welch bound}

\end{axis}
\end{tikzpicture}%
}  

   \caption{{The mutual coherence for different sampling patterns on $\SO(3)$}}
    \label{Coherence_comparison_Wigner_equiang}
\end{figure}

A concern about our pattern search algorithm is computational complexity.  For $N=49$ and $N=100$ and the error tolerance of $\card{\mu(\vec{\theta},\vec{\phi}_k) - \mu_{\mrm{LB}}} \leq \epsilon = 10^{-4}$, the computation time of the algorithm is shown in  Figure \ref{computation_time}.
When we double the dimension of the signal, it is apparent that the computation time to achieve the same error tolerance would increase approximately fivefold.

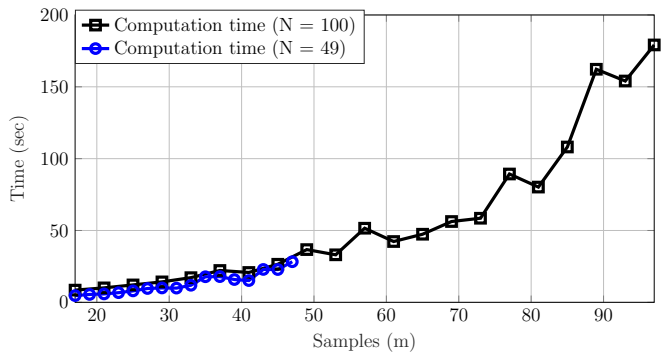
\begin{figure}[!htb]
  \centering
     \scalebox{0.6}{
%
%
\begin{tikzpicture}

\begin{axis}[%
width=5in,
height=2.5in,
at={(3.338in,1.74in)},
scale only axis,
xmin=17,
xmax=97,
xlabel style={font=\color{white!15!black}},
xlabel={Samples (m)},
ymin=0,
ymax=200,
ylabel style={font=\color{white!15!black}},
ylabel={Time (sec)},
axis background/.style={fill=white},
xmajorgrids,
ymajorgrids,
legend style={at={(0,.83)}, anchor=south west, legend cell align=left, align=left, draw=white!15!black}
]
\addplot [color=black, line width=2.0pt, mark size=3pt, mark=square, mark options={solid, black}]
  table[row sep=crcr]{%
5	0\\
9	0\\
13	0\\
17	8.493512\\
21	10.060559\\
25	12.117653\\
29	14.232518\\
33	17.17216\\
37	22.165602\\
41	20.778686\\
45	26.489834\\
49	36.705322\\
53	33.073001\\
57	51.583234\\
61	42.32053\\
65	47.395663\\
69	56.243713\\
73	58.522683\\
77	89.314563\\
81	80.213985\\
85	108.108168\\
89	162.215215\\
93	154.032785\\
97	179.139181\\
};
\addlegendentry{Computation time (N = 100)}

\addplot [color=blue, line width=2.0pt, mark size= 3.0pt, mark=o, mark options={solid, blue}]
  table[row sep=crcr]{%
5	1.644011\\
7	1.956754\\
9	2.439007\\
11	2.845903\\
13	3.553987\\
15	3.919432\\
17	4.6719\\
19	5.482128\\
21	5.844254\\
23	6.719725\\
25	7.953898\\
27	9.623431\\
29	9.999718\\
31	9.896797\\
33	11.880209\\
35	17.802003\\
37	17.892676\\
39	15.859727\\
41	15.150455\\
43	22.88695\\
45	22.817754\\
47	28.259154\\
};
\addlegendentry{Computation time (N = 49)}

\end{axis}
\end{tikzpicture}%
}  

   \caption{{Computation time of algorithm \ref{algo_ps}}}
    \label{computation_time}
\end{figure}


\section{Experimental Results}\label{Sect:Experiments}

In the previous section, we designed two equispaced sampling patterns, one for the sphere and one for the rotation group with better mutual coherence. In this section, we see if this superiority is translated to the sparse recovery performance as well. Besides, the performance of our proposed sampling patterns is compared with random sampling patterns, which are provably good with high probability for sparse recovery. Two random sampling patterns are considered. The first one is proposed in \cite{rauhut2011sparse} with the uniform measure, i.e., $\d\nu=\d\theta\d \phi $ for $\S^2$ and $\d\nu=\d\theta\d \phi \d\chi$ for $\SO(3)$. The second one is given in \cite{burq2012weighted} with the measure
$\d\nu=|\tan \theta|^{1/3}\d\theta\d \phi$ for $\S^2$ and $\d\nu=|\tan \theta|^{1/3}\d\theta\d \phi \d\chi$ for $\SO(3)$.

\subsection{Phase transition diagrams}

Consider the span of band-limited spherical harmonics with $B=10$, that is $N = B^2 = 100$. We use the  equispaced sampling pattern with $\theta_p$ as $\cos \theta_{p} = \frac{2p-m-1}{m-1},\, p \in [m]$, and the azimuth samples $\phi_p$ chosen from Algorithm \ref{algo_ps}. We solve the linear inverse problem without additive noise using the $l_1$-norm minimization package YALL1 \cite{zhang2010yall1}. The phase transition diagram of our proposed sampling pattern is plotted with 50 trials and error threshold  $10^{-3}$. Figure \ref{phase_trans_SH} compares the recovery performance of the proposed sampling pattern with several well-known sampling patterns on the sphere and, as well, random sampling. Not only our proposed sampling gives better recovery performance compared with many regular sampling patterns, it even gives a slightly better sparse recovery performance compared with the two random sampling patterns.

\begin{figure}[!htb]
  \centering
     \scalebox{0.6}{
%
%
\definecolor{mycolor1}{rgb}{1.00000,0.00000,1.00000}%
\definecolor{mycolor2}{rgb}{0.85000,0.32500,0.09800}%
\definecolor{mycolor3}{rgb}{0.92900,0.69400,0.12500}%
\definecolor{mycolor4}{rgb}{0.49400,0.18400,0.55600}%
\definecolor{mycolor5}{rgb}{0.46600,0.67400,0.18800}%
\definecolor{mycolor6}{rgb}{0.30100,0.74500,0.93300}%
\definecolor{mycolor7}{rgb}{0.63500,0.07800,0.18400}%
\definecolor{mycolor8}{rgb}{0.00000,0.44700,0.74100}%
\begin{tikzpicture}

\begin{axis}[%
width=5in,
height=2.51in,
at={(3.237in,1.551in)},
scale only axis,
point meta min=0,
point meta max=1,
axis on top,
xmin=0.08,
xmax=0.98,
xlabel style={font=\color{white!15!black}},
xlabel={m/N},
ymin=-0.05,
ymax=1.05,
ylabel style={font=\color{white!15!black}},
ylabel={s/m},
axis background/.style={fill=white},
title style={font=\bfseries},
legend style={at={(0,0.44)}, anchor=south west, legend cell align=left, align=left, draw=white!15!black},
colormap/blackwhite,
colorbar horizontal,
colorbar style={ xlabel= Success Rate}
]
\addplot [forget plot] graphics [xmin=0.035, xmax=1.025, ymin=-0.05, ymax=1.05] {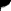};
\addplot [color=mycolor1, line width=2.0pt]
  table[row sep=crcr]{%
0.08	0.0581818181818182\\
0.14	0.218181818181818\\
0.18	0.247272727272727\\
0.26	0.250909090909091\\
0.32	0.221818181818182\\
0.42	0.178181818181818\\
0.5	0.158181818181818\\
0.62	0.134545454545455\\
0.72	0.118181818181818\\
0.86	0.103636363636364\\
0.98	0.101818181818182\\
};
\addlegendentry{Equiangular (Thm. \ref{thm:fullcoherence})}

\addplot [color=mycolor2, line width=2.0pt]
  table[row sep=crcr]{%
0.08	0.0745454545454545\\
0.14	0.305454545454545\\
0.18	0.316363636363636\\
0.26	0.352727272727273\\
0.32	0.365454545454545\\
0.42	0.374545454545454\\
0.5	0.389090909090909\\
0.62	0.414545454545454\\
0.72	0.421818181818182\\
0.86	0.48\\
0.98	0.683636363636364\\
};
\addlegendentry{Spiral \cite{saff1997distributing}}

\addplot [color=mycolor3, line width=2.0pt]
  table[row sep=crcr]{%
0.08	0.301818181818182\\
0.14	0.296363636363636\\
0.18	0.323636363636364\\
0.26	0.32\\
0.32	0.341818181818182\\
0.42	0.365454545454545\\
0.5	0.363636363636364\\
0.62	0.396363636363636\\
0.72	0.421818181818182\\
0.86	0.518181818181818\\
0.98	0.84\\
};
\addlegendentry{Fibonacci \cite{swinbank2006fibonacci}}

\addplot [color=mycolor4, line width=2.0pt]
  table[row sep=crcr]{%
0.08	0.0909090909090909\\
0.14	0.236363636363636\\
0.18	0.332727272727273\\
0.26	0.376363636363636\\
0.32	0.363636363636364\\
0.42	0.427272727272727\\
0.5	0.449090909090909\\
0.62	0.489090909090909\\
0.72	0.523636363636364\\
0.86	0.62\\
0.98	0.783636363636364\\
};
\addlegendentry{Hammersley \cite{cui1997equidistribution}}

\addplot [color=mycolor5, line width=2.0pt]
  table[row sep=crcr]{%
0.08	0.0563636363636364\\
0.14	0.318181818181818\\
0.18	0.347272727272727\\
0.26	0.372727272727273\\
0.32	0.383636363636364\\
0.42	0.427272727272727\\
0.5	0.470909090909091\\
0.62	0.527272727272727\\
0.72	0.58\\
0.86	0.672727272727273\\
0.98	0.84\\
};
\addlegendentry{Proposed (Alg. \ref{algo_ps})}

\addplot [color=mycolor6, line width=2.0pt]
  table[row sep=crcr]{%
0.08	0.101818181818182\\
0.14	0.294545454545455\\
0.18	0.336363636363636\\
0.26	0.38\\
0.32	0.381818181818182\\
0.42	0.42\\
0.5	0.430909090909091\\
0.62	0.447272727272727\\
0.72	0.470909090909091\\
0.86	0.563636363636364\\
0.98	0.774545454545455\\
};
\addlegendentry{t-design \cite{womersley2018efficient}}

\addplot [color=mycolor7, line width=2.0pt]
  table[row sep=crcr]{%
0.08	0.123636363636364\\
0.14	0.305454545454545\\
0.18	0.325454545454545\\
0.26	0.34\\
0.32	0.370909090909091\\
0.42	0.4\\
0.5	0.438181818181818\\
0.62	0.476363636363636\\
0.72	0.512727272727273\\
0.86	0.578181818181818\\
0.98	0.656363636363636\\
};
\addlegendentry{Random \cite{rauhut2011sparse}}

\addplot [color=mycolor8, line width=2.0pt]
  table[row sep=crcr]{%
0.08	0.265454545454545\\
0.14	0.327272727272727\\
0.18	0.369090909090909\\
0.26	0.365454545454545\\
0.32	0.378181818181818\\
0.42	0.443636363636364\\
0.5	0.467272727272727\\
0.62	0.518181818181818\\
0.72	0.558181818181818\\
0.86	0.649090909090909\\
0.98	0.805454545454546\\
};
\addlegendentry{Random \cite{burq2012weighted}}

\end{axis}
\end{tikzpicture}%
}  

   \caption{{Phase transition diagram of different sampling patterns on the sphere}}
    \label{phase_trans_SH}
\end{figure}
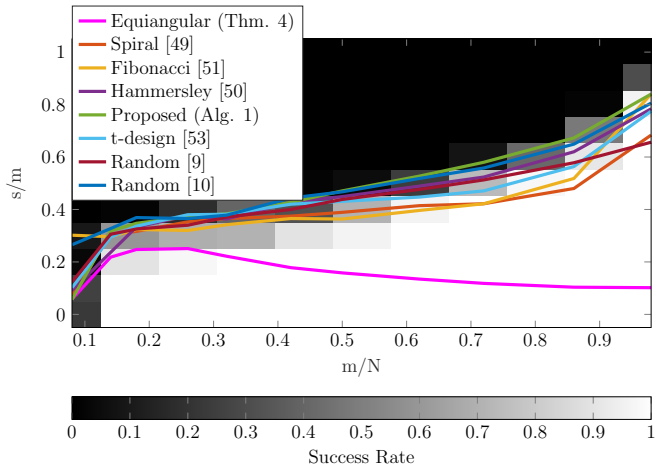

A similar result is observed for Wigner D-functions. We consider band-limited functions with $B = 4$ and $N =\frac{B(2B-1)(2B+1)}{3} = 84$. Figure \ref{phase_trans_Wigner_equiang} shows the phase transition for the $l_1$-minimization. Although our proposed sampling pattern for Wigner D-functions does not achieve the lower bound, it still outperforms other regular sampling patterns, and even slightly random sampling patterns.
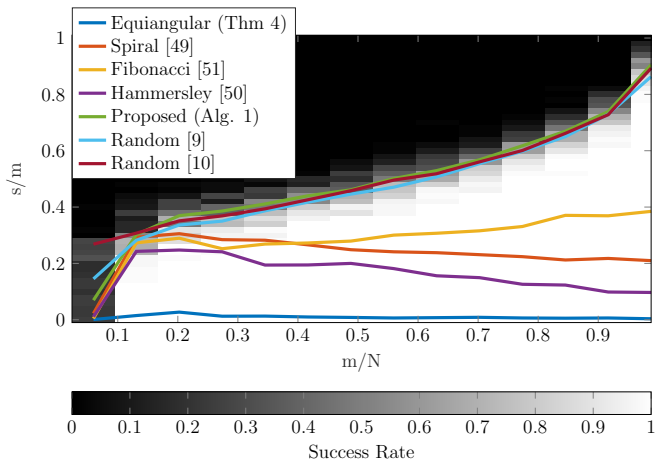
\begin{figure}[!htb]
  \centering
     \scalebox{0.6}{
%
%
\definecolor{mycolor1}{rgb}{0.00000,0.44700,0.74100}%
\definecolor{mycolor2}{rgb}{0.85000,0.32500,0.09800}%
\definecolor{mycolor3}{rgb}{0.92900,0.69400,0.12500}%
\definecolor{mycolor4}{rgb}{0.49400,0.18400,0.55600}%
\definecolor{mycolor5}{rgb}{0.46600,0.67400,0.18800}%
\definecolor{mycolor6}{rgb}{0.30100,0.74500,0.93300}%
\definecolor{mycolor7}{rgb}{0.63500,0.07800,0.18400}%
\begin{tikzpicture}

\begin{axis}[%
width=5in,
height=2.5in,
at={(3.255in,1.735in)},
scale only axis,
point meta min=0,
point meta max=1,
axis on top,
xmin=0.0238095238095238,
xmax=0.9881,
xlabel style={font=\color{white!15!black}},
xlabel={m/N},
ymin=-0.01,
ymax=1.01,
ylabel style={font=\color{white!15!black}},
ylabel={s/m},
axis background/.style={fill=white},
legend style={at={(0,0.5)}, anchor=south west, legend cell align=left, align=left, draw=white!15!black},
colormap/blackwhite,
colorbar horizontal,
colorbar style={ xlabel= Success Rate}
]
\addplot [forget plot] graphics [xmin=0.0238095238095238, xmax=1.02380952380952, ymin=-0.01, ymax=1.01] {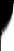};
\addplot [color=mycolor1, line width=2.0pt]
  table[row sep=crcr]{%
0.0595238095238095	0.000784313725490196\\
0.130952380952381	0.0149019607843137\\
0.202380952380952	0.0270588235294118\\
0.273809523809524	0.0125490196078431\\
0.345238095238095	0.0129411764705882\\
0.416666666666667	0.00980392156862745\\
0.488095238095238	0.00823529411764706\\
0.55952380952381	0.00627450980392157\\
0.630952380952381	0.00745098039215686\\
0.702380952380952	0.00862745098039216\\
0.773809523809524	0.00627450980392157\\
0.845238095238095	0.00549019607843137\\
0.916666666666667	0.00627450980392157\\
0.988095238095238	0.00352941176470588\\
};
\addlegendentry{Equiangular (Thm \ref{thm:fullcoherence})}

\addplot [color=mycolor2, line width=2.0pt]
  table[row sep=crcr]{%
0.0595238095238095	0.0223529411764706\\
0.130952380952381	0.292941176470588\\
0.202380952380952	0.305490196078431\\
0.273809523809524	0.284313725490196\\
0.345238095238095	0.282352941176471\\
0.416666666666667	0.265882352941176\\
0.488095238095238	0.248627450980392\\
0.55952380952381	0.241176470588235\\
0.630952380952381	0.238039215686275\\
0.702380952380952	0.230980392156863\\
0.773809523809524	0.223921568627451\\
0.845238095238095	0.212156862745098\\
0.916666666666667	0.217647058823529\\
0.988095238095238	0.209803921568627\\
};
\addlegendentry{Spiral \cite{saff1997distributing}}

\addplot [color=mycolor3, line width=2.0pt]
  table[row sep=crcr]{%
0.0595238095238095	0.00431372549019608\\
0.130952380952381	0.272941176470588\\
0.202380952380952	0.289411764705882\\
0.273809523809524	0.252156862745098\\
0.345238095238095	0.268627450980392\\
0.416666666666667	0.272549019607843\\
0.488095238095238	0.27921568627451\\
0.55952380952381	0.300392156862745\\
0.630952380952381	0.307058823529412\\
0.702380952380952	0.315686274509804\\
0.773809523809524	0.330980392156863\\
0.845238095238095	0.370588235294118\\
0.916666666666667	0.369019607843137\\
0.988095238095238	0.384705882352941\\
};
\addlegendentry{Fibonacci \cite{swinbank2006fibonacci}}

\addplot [color=mycolor4, line width=2.0pt]
  table[row sep=crcr]{%
0.0595238095238095	0.0105882352941176\\
0.130952380952381	0.242745098039216\\
0.202380952380952	0.247450980392157\\
0.273809523809524	0.241176470588235\\
0.345238095238095	0.194117647058823\\
0.416666666666667	0.194509803921569\\
0.488095238095238	0.2\\
0.55952380952381	0.18156862745098\\
0.630952380952381	0.156078431372549\\
0.702380952380952	0.149411764705882\\
0.773809523809524	0.125882352941176\\
0.845238095238095	0.123137254901961\\
0.916666666666667	0.0984313725490196\\
0.988095238095238	0.0968627450980392\\
};
\addlegendentry{Hammersley \cite{cui1997equidistribution}}

\addplot [color=mycolor5, line width=2.0pt]
  table[row sep=crcr]{%
0.0595238095238095	0.0694117647058824\\
0.130952380952381	0.302352941176471\\
0.202380952380952	0.369411764705882\\
0.273809523809524	0.385882352941176\\
0.345238095238095	0.410980392156863\\
0.416666666666667	0.44\\
0.488095238095238	0.461176470588235\\
0.55952380952381	0.501176470588235\\
0.630952380952381	0.529411764705882\\
0.702380952380952	0.567843137254902\\
0.773809523809524	0.617647058823529\\
0.845238095238095	0.669019607843137\\
0.916666666666667	0.738823529411765\\
0.988095238095238	0.905098039215686\\
};
\addlegendentry{Proposed (Alg. \ref{algo_ps})}

\addplot [color=mycolor6, line width=2.0pt]
  table[row sep=crcr]{%
0.0595238095238095	0.145098039215686\\
0.130952380952381	0.282352941176471\\
0.202380952380952	0.336078431372549\\
0.273809523809524	0.350196078431373\\
0.345238095238095	0.386666666666666\\
0.416666666666667	0.417254901960784\\
0.488095238095238	0.445882352941177\\
0.55952380952381	0.470588235294117\\
0.630952380952381	0.507058823529412\\
0.702380952380952	0.554117647058823\\
0.773809523809524	0.598039215686275\\
0.845238095238095	0.649803921568627\\
0.916666666666667	0.729411764705882\\
0.988095238095238	0.863137254901961\\
};
\addlegendentry{Random \cite{rauhut2011sparse}}

\addplot [color=mycolor7, line width=2.0pt]
  table[row sep=crcr]{%
0.0595238095238095	0.268235294117647\\
0.130952380952381	0.307450980392157\\
0.202380952380952	0.350980392156863\\
0.273809523809524	0.368627450980392\\
0.345238095238095	0.393725490196078\\
0.416666666666667	0.425882352941176\\
0.488095238095238	0.458039215686274\\
0.55952380952381	0.495686274509804\\
0.630952380952381	0.518431372549019\\
0.702380952380952	0.560392156862745\\
0.773809523809524	0.60078431372549\\
0.845238095238095	0.661568627450981\\
0.916666666666667	0.727843137254902\\
0.988095238095238	0.892941176470588\\
};
\addlegendentry{Random \cite{burq2012weighted}}

\end{axis}
\end{tikzpicture}%
}  

   \caption{{Phase transition diagram of different sampling patterns on the rotation group}}
    \label{phase_trans_Wigner_equiang}
\end{figure}
{The comparison between several recovery algorithms is presented in Figure \ref{phase_trans_algorithm} for Wigner D-functions, where besides \ac{BP}, the \ac{OMP} \cite{tropp2007signal} and the \ac{AMP} \cite{donoho2009message} are also implemented.
It can be seen that the proposed sampling pattern performs slightly better than the random sampling.
Furthermore, \ac{OMP} algorithm delivers better recovery in this case. In this case, the sparsity $s = 20$ is considered and the non-zero values are drawn from random zero mean and unit variance Gaussian distribution. Signal recoveries are conducted with $30$ trials.}
\begin{figure}[!htb]
  \centering
     \scalebox{0.6}{
%
%
\begin{tikzpicture}

\begin{axis}[%
width=5in,
height=2.5in,
at={(3.467in,1.529in)},
scale only axis,
xmin=0,
xmax=1,
xlabel style={font=\color{white!15!black}},
xlabel={m/N},
ymin=0,
ymax=1,
ylabel style={font=\color{white!15!black}},
ylabel={Success Rate},
axis background/.style={fill=white},
xmajorgrids,
ymajorgrids,
legend style={at={(0,0.37)}, anchor=south west,legend cell align=left, align=left, draw=white!15!black}
]
\addplot [color=blue,line width=2.0pt, dashed]
  table[row sep=crcr]{%
0.0595238095238095	0\\
0.130952380952381	0\\
0.202380952380952	0\\
0.273809523809524	0\\
0.345238095238095	0\\
0.416666666666667	0\\
0.488095238095238	0.233333333333333\\
0.55952380952381	0.833333333333333\\
0.630952380952381	1\\
0.702380952380952	1\\
0.773809523809524	1\\
0.845238095238095	1\\
0.916666666666667	1\\
0.988095238095238	1\\
};
\addlegendentry{Random \cite{rauhut2011sparse} BP}

\addplot [color=red,line width=2.0pt, dashed]
  table[row sep=crcr]{%
0.0595238095238095	0\\
0.130952380952381	0\\
0.202380952380952	0\\
0.273809523809524	0\\
0.345238095238095	0\\
0.416666666666667	0\\
0.488095238095238	0.233333333333333\\
0.55952380952381	0.933333333333333\\
0.630952380952381	1\\
0.702380952380952	1\\
0.773809523809524	1\\
0.845238095238095	1\\
0.916666666666667	1\\
0.988095238095238	1\\
};
\addlegendentry{Random \cite{burq2012weighted} BP}

\addplot [color=black,line width=2.0pt, dashed]
  table[row sep=crcr]{%
0.0595238095238095	0\\
0.130952380952381	0\\
0.202380952380952	0\\
0.273809523809524	0\\
0.345238095238095	0\\
0.416666666666667	0\\
0.488095238095238	0.166666666666667\\
0.55952380952381	0.933333333333333\\
0.630952380952381	1\\
0.702380952380952	1\\
0.773809523809524	1\\
0.845238095238095	1\\
0.916666666666667	1\\
0.988095238095238	1\\
};
\addlegendentry{Proposed (Alg.\ref{algo_ps}) BP}

\addplot [color=blue,line width=2.0pt, dashed, mark=square, mark size=3pt, mark options={solid, blue}]
  table[row sep=crcr]{%
0.0595238095238095	0\\
0.130952380952381	0\\
0.202380952380952	0\\
0.273809523809524	0\\
0.345238095238095	0\\
0.416666666666667	0\\
0.488095238095238	0.133333333333333\\
0.55952380952381	0.733333333333333\\
0.630952380952381	0.866666666666667\\
0.702380952380952	0.966666666666667\\
0.773809523809524	1\\
0.845238095238095	1\\
0.916666666666667	1\\
0.988095238095238	1\\
};
\addlegendentry{Random \cite{rauhut2011sparse} OMP}

\addplot [color=red,line width=2.0pt, dashed, mark=square, mark size=3pt, mark options={solid, red}]
  table[row sep=crcr]{%
0.0595238095238095	0\\
0.130952380952381	0\\
0.202380952380952	0\\
0.273809523809524	0\\
0.345238095238095	0\\
0.416666666666667	0.2\\
0.488095238095238	0.433333333333333\\
0.55952380952381	0.966666666666667\\
0.630952380952381	1\\
0.702380952380952	1\\
0.773809523809524	1\\
0.845238095238095	1\\
0.916666666666667	1\\
0.988095238095238	1\\
};
\addlegendentry{Random \cite{burq2012weighted} OMP}

\addplot [color=black,line width=2.0pt, dashed, mark=square,mark size=3pt, mark options={solid, black}]
  table[row sep=crcr]{%
0.0595238095238095	0\\
0.130952380952381	0\\
0.202380952380952	0\\
0.273809523809524	0\\
0.345238095238095	0\\
0.416666666666667	0.3\\
0.488095238095238	0.866666666666667\\
0.55952380952381	1\\
0.630952380952381	1\\
0.702380952380952	1\\
0.773809523809524	1\\
0.845238095238095	1\\
0.916666666666667	1\\
0.988095238095238	1\\
};
\addlegendentry{Proposed (Alg.\ref{algo_ps}) OMP}

\addplot [color=blue,line width=2.0pt, dashed, mark=triangle,mark size=3pt, mark options={solid, blue}]
  table[row sep=crcr]{%
0.0595238095238095	0\\
0.130952380952381	0\\
0.202380952380952	0\\
0.273809523809524	0\\
0.345238095238095	0\\
0.416666666666667	0\\
0.488095238095238	0.0333333333333333\\
0.55952380952381	0.233333333333333\\
0.630952380952381	0.6\\
0.702380952380952	0.966666666666667\\
0.773809523809524	1\\
0.845238095238095	1\\
0.916666666666667	1\\
0.988095238095238	0.966666666666667\\
};
\addlegendentry{Random \cite{rauhut2011sparse} AMP}

\addplot [color=red,line width=2.0pt, mark=triangle,mark size=3pt, mark options={solid, red}]
  table[row sep=crcr]{%
0.0595238095238095	0\\
0.130952380952381	0\\
0.202380952380952	0\\
0.273809523809524	0\\
0.345238095238095	0\\
0.416666666666667	0\\
0.488095238095238	0.0333333333333333\\
0.55952380952381	0.466666666666667\\
0.630952380952381	0.866666666666667\\
0.702380952380952	1\\
0.773809523809524	1\\
0.845238095238095	1\\
0.916666666666667	1\\
0.988095238095238	1\\
};
\addlegendentry{Random \cite{burq2012weighted} AMP}

\addplot [color=black,line width=2.0pt, mark=triangle,mark size=3pt, mark options={solid, black}]
  table[row sep=crcr]{%
0.0595238095238095	0\\
0.130952380952381	0\\
0.202380952380952	0\\
0.273809523809524	0\\
0.345238095238095	0\\
0.416666666666667	0\\
0.488095238095238	0.166666666666667\\
0.55952380952381	0.9\\
0.630952380952381	1\\
0.702380952380952	1\\
0.773809523809524	1\\
0.845238095238095	1\\
0.916666666666667	1\\
0.988095238095238	1\\
};
\addlegendentry{Proposed (Alg.\ref{algo_ps}) AMP}

\end{axis}
\end{tikzpicture}%
}  

   \caption{{Phase transition diagram of different algorithms}}
    \label{phase_trans_algorithm}
\end{figure}
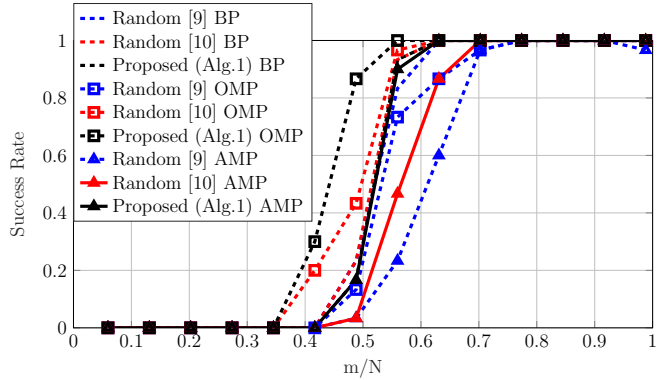

\subsection{Spherical near-field antenna measurements}
One of the main applications of sparse recovery on $\S^2$ and $\SO(3)$ is spherical near-field antenna measurement \cite{hansen1988spherical}. The expression of electromagnetic field of the antenna using Wigner D-basis coefficients is described as follows 
\begin{equation}
y(\theta,\phi,\chi)=v \sum \limits_{n=-v_{max}}^{v_{max}}\sum\limits_{h=1}^2\sum\limits_{l=1}^B\sum\limits_{k=-l}^l T_{hlk} \mathrm D_l^{k,n}(\theta,\phi,\chi)
\end{equation} 
where $y(\theta,\phi,\chi)$ is a band-limited near-field signal with Wigner D-functions as basis, $h$ denotes the both transverse electric (TE) and magnetic (TM), $n$ and $\chi$ denote order and angle to measure polarization, respectively. The bandwidth $B$ is obtained by calculating the wavenumber $k$ and minimum sphere that could cover the whole antenna with radius $r_0$. The bandwidth is given by $B = kr_0 + 10$, where the factor $10$ is usually added as a correction factor. Normally, it is desirable to measure co- and cross-polarization of the antenna and to use $n= \pm 1$, with angle $\chi \in \{0,\pi/2\}$. The goal is to estimate the spherical wave coefficients of the antenna under test, i.e., $T_{hlk}$ in near-field measurements and use it to determine far-field patterns. 

\begin{figure}[!htb]
  \centering
     \scalebox{0.65}{
%
%
%
\begin{tikzpicture}
 
\begin{axis}[%
width=5in,
height=2.5in,
at={(-3in,0in)},
scale only axis,
point meta min=-50,
point meta max=0,
xmin=0,
xmax=2,
xlabel style={font=\color{white!15!black}},
xlabel={k (Order)},
ymin=0,
ymax=2.0 ,
axis line style={draw=none},
 tick style={draw=none},
ylabel style={font=\color{white!15!black}},
ylabel={l (Degree)},
yticklabels={,,},
xticklabels={,,},
ylabel near ticks,
xlabel near ticks,
legend style={at={(0.01,0.45)}, anchor=south west, legend cell align=left, align=left, draw=white!15!black},
colormap={mymap}{[1pt] rgb(0pt)=(0,0,0); rgb(1pt)=(0.211624,0.189781,0.577676); rgb(2pt)=(0.212252,0.213771,0.626971); rgb(3pt)=(0.2081,0.2386,0.677086); rgb(4pt)=(0.195905,0.264457,0.7279); rgb(5pt)=(0.170729,0.291938,0.779248); rgb(6pt)=(0.125271,0.324243,0.830271); rgb(7pt)=(0.0591333,0.359833,0.868333); rgb(8pt)=(0.0116952,0.38751,0.881957); rgb(9pt)=(0.00595714,0.408614,0.882843); rgb(10pt)=(0.0165143,0.4266,0.878633); rgb(11pt)=(0.0328524,0.443043,0.871957); rgb(12pt)=(0.0498143,0.458571,0.864057); rgb(13pt)=(0.0629333,0.47369,0.855438); rgb(14pt)=(0.0722667,0.488667,0.8467); rgb(15pt)=(0.0779429,0.503986,0.838371); rgb(16pt)=(0.0793476,0.520024,0.831181); rgb(17pt)=(0.0749429,0.537543,0.826271); rgb(18pt)=(0.0640571,0.556986,0.823957); rgb(19pt)=(0.0487714,0.577224,0.822829); rgb(20pt)=(0.0343429,0.596581,0.819852); rgb(21pt)=(0.0265,0.6137,0.8135); rgb(22pt)=(0.0238905,0.628662,0.803762); rgb(23pt)=(0.0230905,0.641786,0.791267); rgb(24pt)=(0.0227714,0.653486,0.776757); rgb(25pt)=(0.0266619,0.664195,0.760719); rgb(26pt)=(0.0383714,0.674271,0.743552); rgb(27pt)=(0.0589714,0.683757,0.725386); rgb(28pt)=(0.0843,0.692833,0.706167); rgb(29pt)=(0.113295,0.7015,0.685857); rgb(30pt)=(0.145271,0.709757,0.664629); rgb(31pt)=(0.180133,0.717657,0.642433); rgb(32pt)=(0.217829,0.725043,0.619262); rgb(33pt)=(0.258643,0.731714,0.595429); rgb(34pt)=(0.302171,0.737605,0.571186); rgb(35pt)=(0.348167,0.742433,0.547267); rgb(36pt)=(0.395257,0.7459,0.524443); rgb(37pt)=(0.44201,0.748081,0.503314); rgb(38pt)=(0.487124,0.749062,0.483976); rgb(39pt)=(0.530029,0.749114,0.466114); rgb(40pt)=(0.570857,0.748519,0.44939); rgb(41pt)=(0.609852,0.747314,0.433686); rgb(42pt)=(0.6473,0.7456,0.4188); rgb(43pt)=(0.683419,0.743476,0.404433); rgb(44pt)=(0.71841,0.741133,0.390476); rgb(45pt)=(0.752486,0.7384,0.376814); rgb(46pt)=(0.785843,0.735567,0.363271); rgb(47pt)=(0.818505,0.732733,0.34979); rgb(48pt)=(0.850657,0.7299,0.336029); rgb(49pt)=(0.882433,0.727433,0.3217); rgb(50pt)=(0.913933,0.725786,0.306276); rgb(51pt)=(0.944957,0.726114,0.288643); rgb(52pt)=(0.973895,0.731395,0.266648); rgb(53pt)=(0.993771,0.745457,0.240348); rgb(54pt)=(0.999043,0.765314,0.216414); rgb(55pt)=(0.995533,0.786057,0.196652); rgb(56pt)=(0.988,0.8066,0.179367); rgb(57pt)=(0.978857,0.827143,0.163314); rgb(58pt)=(0.9697,0.848138,0.147452); rgb(59pt)=(0.962586,0.870514,0.1309); rgb(60pt)=(0.958871,0.8949,0.113243); rgb(61pt)=(0.959824,0.921833,0.0948381); rgb(62pt)=(0.9661,0.951443,0.0755333); rgb(63pt)=(0.9763,0.9831,0.0538)},
colorbar horizontal,
colorbar style={at={(0 ,-1.5cm )}, xlabel= Magnitude (dB)}
]
\addplot [forget plot] graphics [xmin=0, xmax=2, ymin=0, ymax=2]{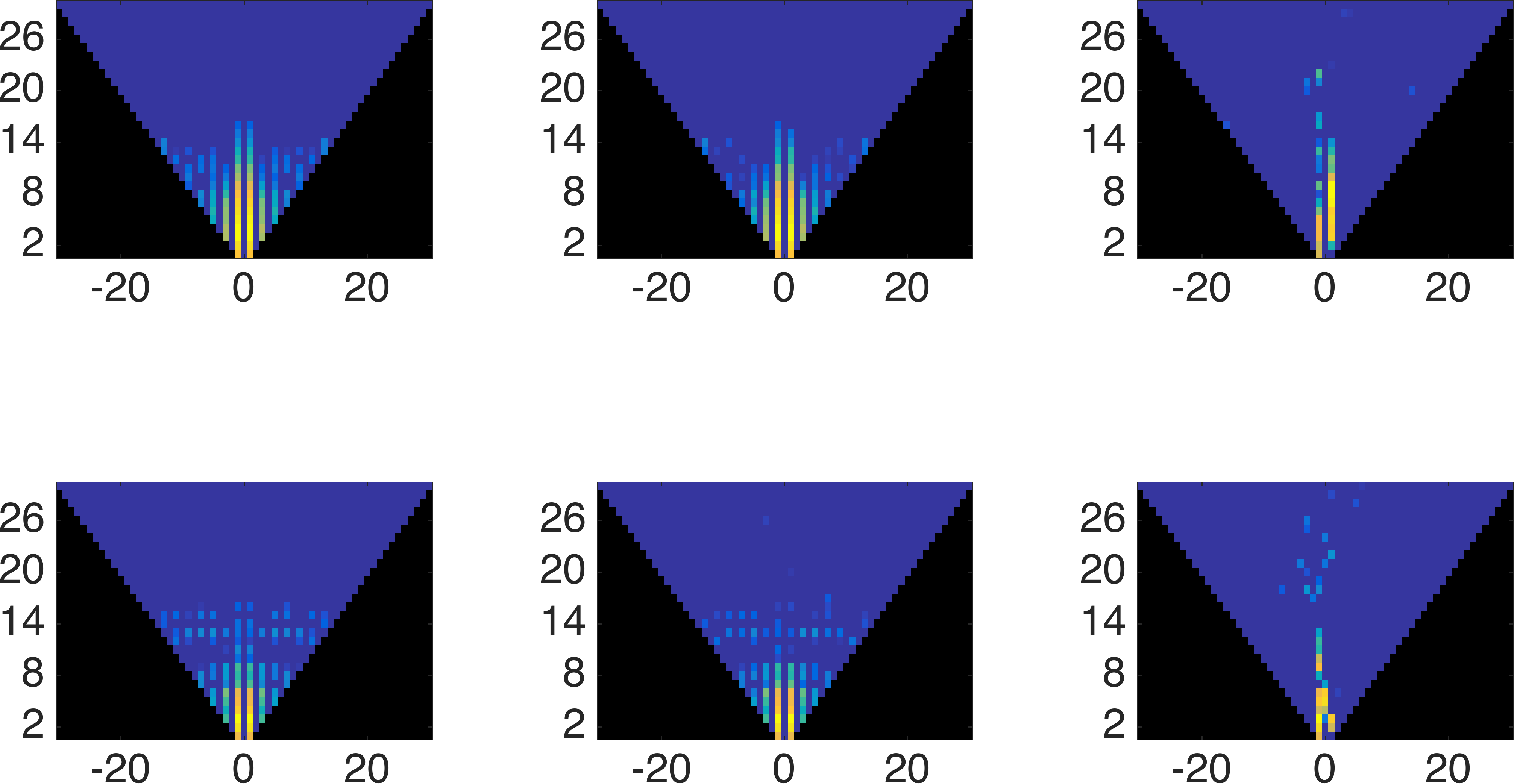};
  
\end{axis}
\begin{scope}[every node/.style={text width=7cm,align=left}]
\node [below,align=left  ,xshift = -3.5cm ,yshift = 1.3cm] at (samples.south){(d) Original (h=2)};
\node [below,align=center,xshift = -1cm  ,yshift =1.3cm] at (samples.south){(e) Proposed (h=2)};
\node [below,align=right ,xshift = 1.5cm,yshift = 1.3cm] at (samples.south){(f) Equiangular (h=2)};

\node [below,align=left  ,xshift = -3.5cm ,yshift = 5.3cm] at (samples.south){(a) Original (h=1)};
\node [below,align=center,xshift = -1cm  ,yshift = 5.3cm] at (samples.south){(b) Proposed (h=1)};
\node [below,align=right ,xshift = 1.5cm,yshift = 5.3cm] at (samples.south){(c) Equiangular (h=1)};

\node [below,align=center,xshift = -1.3 cm  ,yshift = 5.8cm] at (samples.south){k (Order)};

\end{scope}

\end{tikzpicture}%
}  

   \caption{{The original spherical wave coefficients compared with  basis pursuit recovered coefficients from our proposed and equiangular sampling pattern on the sphere}}
    \label{SWE_BP}
\end{figure}
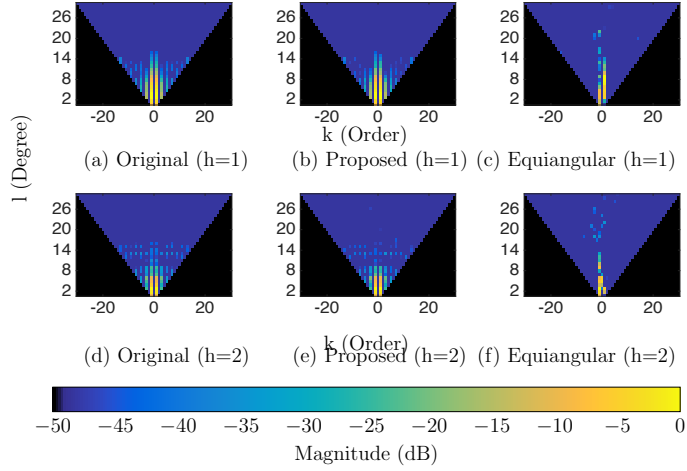

The classical method \cite{hansen1988spherical} uses Fourier analysis with equiangular samples to get the spherical wave coefficient $T_{hlk}$ and lacks the freedom to choose different sampling patterns. 
In the real measurement systems,  the measurement time directly scales with the  number of required samples.

\begin{figure}[!htb]
  \centering
     \scalebox{0.6}{\input{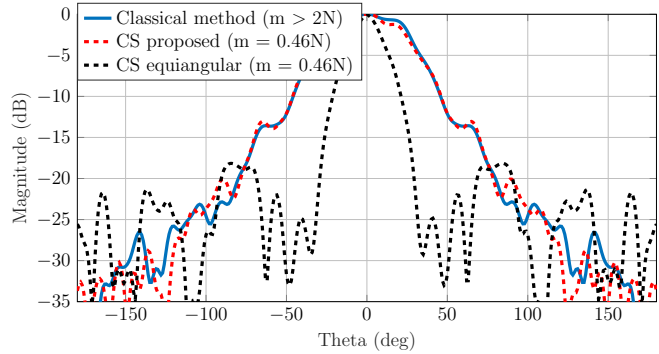}}  

   \caption{{Far-field pattern antenna horn SAS-571 $\phi$-cut = $180^\circ$ and $\chi = 0^\circ$}}
    \label{FF}
\end{figure}

In the classical method, we have to take $m \geq 2(B+1)(2B+1)$. The spherical wave coefficients, however, are sparse with respect to Wigner D-basis, which calls for compressed sensing methods. It can be seen in Figure \ref{SWE_BP} that the important spherical wave coefficients, which is represented by the high intensity of the amplitude, are compressible. In order to get better understanding of spherical near-field measurements we refer to \cite{hansen1988spherical,cornelius2015investigation}. Figure \ref{SWE_BP} shows the estimation of spherical wave coefficients by using basis pursuit for antenna horn SAS-571. The bandwidth in this case is given by $B=30$, which means $N=960$. Note that the number of coefficients are twice this number, namely 1920, because of the TE/TM coefficient $h$.
It can be seen that the proposed sampling manages to recover same spherical wave coefficients as the conventional method with smaller number of measurements, namely $m=900$. As it has been shown in \cite{culotta2018compressed}, our proposed sampling pattern can be used to obtain a smooth trajectory for robotic measurements over the sphere.

The equiangular sampling pattern fails to estimate the spherical wave coefficients. This can be seen as well in far-field signal reconstructions. Figure \ref{FF} and Figure \ref{FF2} show this for polarization $\chi = 0^\circ$ and $\chi = 90^\circ$, respectively. It is assumed that the classical method gives a very good approximation of the ground truth. In comparison, our sampling patterns matches closely the output of the classical method while the equiangular sampling pattern fails to reconstruct the far-field. 
\begin{figure}[!htb]
  \centering
     \scalebox{0.6}{\input{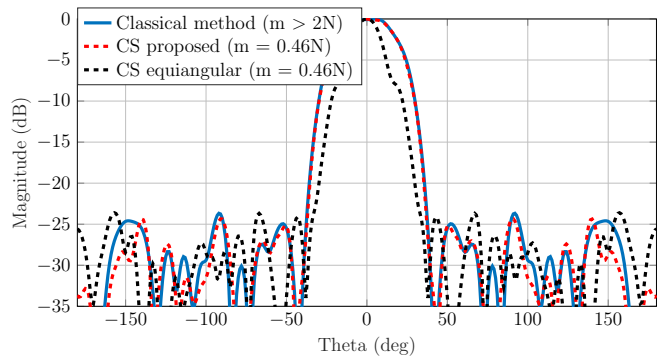}}  

   \caption{{Far-field pattern antenna horn SAS-571 $\phi$-cut = $90^\circ$ and $\chi = 90^\circ$}}
    \label{FF2}
\end{figure}

{\subsection{Earth magnetic fields}
It is also possible to apply our proposed sampling points to the \ac{IGRF} model. \ac{IGRF} model uses the gradient of magnetic scalar potential to describe the earth's geomagnetic field and it can be expressed by using spherical harmonics expansion as
\begin{equation*}
\begin{aligned}
a \sum_{l=1}^{B-1} \sum_{k=0}^{l} \bigg(\frac{a}{r} \bigg)^{l+1} \bigg(g_{l}^{k}(t) \cos k\phi + h_{l}^{k}(t) \sin k\phi \bigg) P_{l}^{k}(\cos \theta),
\end{aligned}
\end{equation*}
where $a$ is the Earth's radius, $r$ is radial distance from the Earth's center. The time varying Gauss coefficients are given as $g_{l}^{k}(t),h_{l}^{k}(t)$. In this case, $P_{l}^{k}(\cos \theta)$ is the normalized associated Legendre polynomials with degree $l$ and order $k$, where the normalization factor is $(-1)^{k}\sqrt{\frac{2(l-k)!}{(l+k)!}}$. 
In this numerical result, we will consider the $2015$ measurements model \cite{thebault2015international} with band-limited spherical harmonics $B=14$, thus the size of spherical harmonics coefficients is given by $N = 196$. The magnetic field is sampled using equiangular, Hammersley and proposed sampling points with number of samples $m = 53$. From these samples, sparse coefficients of spherical harmonics are estimated by using \ac{BP} and projected into spherical harmonics with fine grid resolution on $\theta \in [-\frac{\pi}{2},\frac{\pi}{2}]$ and $\phi \in [0,2\pi)$. Figure $\ref{IGRF}$ shows the comparison of the original and the reconstruction magnetic fields after projecting the spherical harmonics coefficients to spherical harmonics matrix with fine resolution.
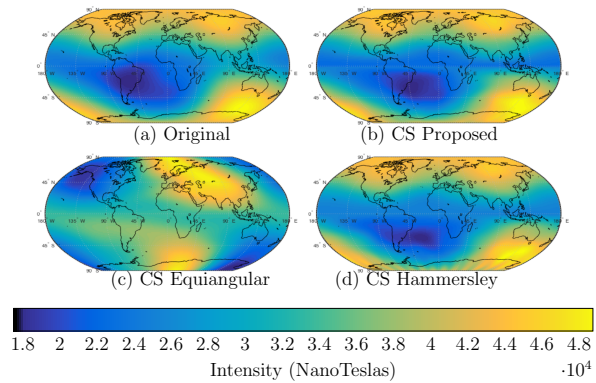
\begin{figure}[!htb]
  \centering
     \scalebox{0.6}{
%
%
%
\begin{tikzpicture}
 
\begin{axis}[%
width=5in,
height=2.5in,
at={(-3in,0in)},
scale only axis,
point meta min=1.7513e+04,
point meta max=4.8746e+04,
xmin=0,
xmax=2,
xlabel style={font=\color{white!15!black}},
ymin=0,
ymax=2.0 ,
axis line style={draw=none},
 tick style={draw=none},
ylabel style={font=\color{white!15!black}},
yticklabels={,,},
xticklabels={,,},
ylabel near ticks,
xlabel near ticks,
legend style={at={(0.01,0.45)}, anchor=south west, legend cell align=left, align=left, draw=white!15!black},
colormap={mymap}{[1pt] rgb(0pt)=(0,0,0); rgb(1pt)=(0.211624,0.189781,0.577676); rgb(2pt)=(0.212252,0.213771,0.626971); rgb(3pt)=(0.2081,0.2386,0.677086); rgb(4pt)=(0.195905,0.264457,0.7279); rgb(5pt)=(0.170729,0.291938,0.779248); rgb(6pt)=(0.125271,0.324243,0.830271); rgb(7pt)=(0.0591333,0.359833,0.868333); rgb(8pt)=(0.0116952,0.38751,0.881957); rgb(9pt)=(0.00595714,0.408614,0.882843); rgb(10pt)=(0.0165143,0.4266,0.878633); rgb(11pt)=(0.0328524,0.443043,0.871957); rgb(12pt)=(0.0498143,0.458571,0.864057); rgb(13pt)=(0.0629333,0.47369,0.855438); rgb(14pt)=(0.0722667,0.488667,0.8467); rgb(15pt)=(0.0779429,0.503986,0.838371); rgb(16pt)=(0.0793476,0.520024,0.831181); rgb(17pt)=(0.0749429,0.537543,0.826271); rgb(18pt)=(0.0640571,0.556986,0.823957); rgb(19pt)=(0.0487714,0.577224,0.822829); rgb(20pt)=(0.0343429,0.596581,0.819852); rgb(21pt)=(0.0265,0.6137,0.8135); rgb(22pt)=(0.0238905,0.628662,0.803762); rgb(23pt)=(0.0230905,0.641786,0.791267); rgb(24pt)=(0.0227714,0.653486,0.776757); rgb(25pt)=(0.0266619,0.664195,0.760719); rgb(26pt)=(0.0383714,0.674271,0.743552); rgb(27pt)=(0.0589714,0.683757,0.725386); rgb(28pt)=(0.0843,0.692833,0.706167); rgb(29pt)=(0.113295,0.7015,0.685857); rgb(30pt)=(0.145271,0.709757,0.664629); rgb(31pt)=(0.180133,0.717657,0.642433); rgb(32pt)=(0.217829,0.725043,0.619262); rgb(33pt)=(0.258643,0.731714,0.595429); rgb(34pt)=(0.302171,0.737605,0.571186); rgb(35pt)=(0.348167,0.742433,0.547267); rgb(36pt)=(0.395257,0.7459,0.524443); rgb(37pt)=(0.44201,0.748081,0.503314); rgb(38pt)=(0.487124,0.749062,0.483976); rgb(39pt)=(0.530029,0.749114,0.466114); rgb(40pt)=(0.570857,0.748519,0.44939); rgb(41pt)=(0.609852,0.747314,0.433686); rgb(42pt)=(0.6473,0.7456,0.4188); rgb(43pt)=(0.683419,0.743476,0.404433); rgb(44pt)=(0.71841,0.741133,0.390476); rgb(45pt)=(0.752486,0.7384,0.376814); rgb(46pt)=(0.785843,0.735567,0.363271); rgb(47pt)=(0.818505,0.732733,0.34979); rgb(48pt)=(0.850657,0.7299,0.336029); rgb(49pt)=(0.882433,0.727433,0.3217); rgb(50pt)=(0.913933,0.725786,0.306276); rgb(51pt)=(0.944957,0.726114,0.288643); rgb(52pt)=(0.973895,0.731395,0.266648); rgb(53pt)=(0.993771,0.745457,0.240348); rgb(54pt)=(0.999043,0.765314,0.216414); rgb(55pt)=(0.995533,0.786057,0.196652); rgb(56pt)=(0.988,0.8066,0.179367); rgb(57pt)=(0.978857,0.827143,0.163314); rgb(58pt)=(0.9697,0.848138,0.147452); rgb(59pt)=(0.962586,0.870514,0.1309); rgb(60pt)=(0.958871,0.8949,0.113243); rgb(61pt)=(0.959824,0.921833,0.0948381); rgb(62pt)=(0.9661,0.951443,0.0755333); rgb(63pt)=(0.9763,0.9831,0.0538)},
colorbar horizontal,
colorbar style={at={(0 ,-0.5cm )}, xlabel= Intensity (NanoTeslas)}
]
\addplot [forget plot] graphics [xmin=0, xmax=2, ymin=0, ymax=2]{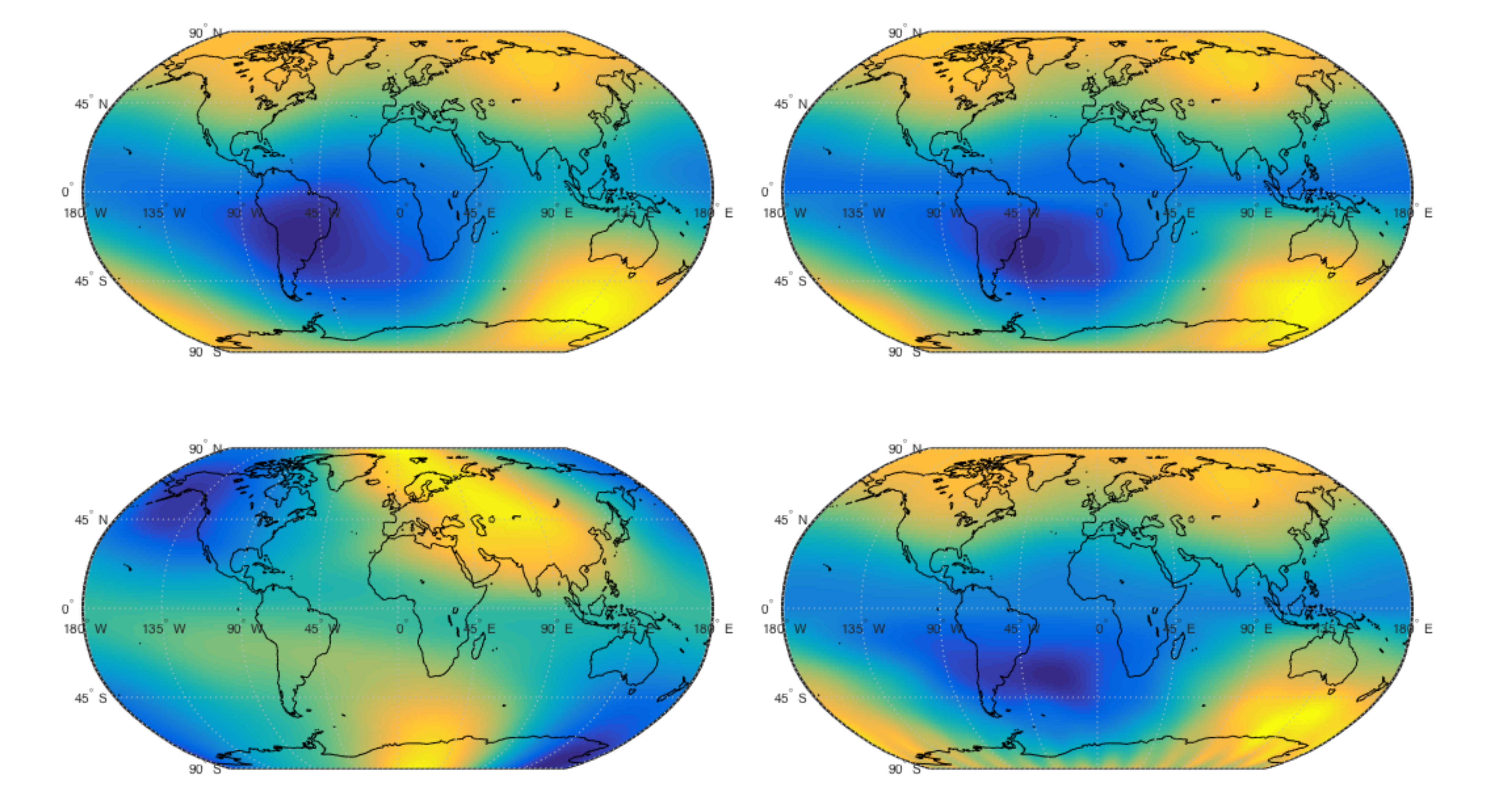};
  
\end{axis}
\begin{scope}[every node/.style={text width=7cm,align=left}]
\node [below,align=left  ,xshift = -2 cm ,yshift = 2.3cm] at (samples.south){(c) CS Equiangular };
\node [below,align=right ,xshift = -0.5 cm,yshift = 2.3cm] at (samples.south){(d) CS Hammersley};

\node [below,align=left  ,xshift = -1.5cm ,yshift = 5.5cm] at (samples.south){(a) Original};
\node [below,align=right ,xshift =  -0.5cm,yshift = 5.5cm] at (samples.south){(b) CS Proposed};


\end{scope}

\end{tikzpicture}%
}  

   \caption{Reconstruction of the earth magnetic field by using \ac{CS} with different sampling patterns}
    \label{IGRF}
\end{figure}
It can be seen the proposed sampling points perform slightly better reconstruction than Hammersley sampling points to reconstruct the earth's magnetic field by using \ac{BP}. As mentioned earlier, the equiangular sampling points deliver the worst reconstruction among the sampling points. }
 
\section{Conclusion and Future Works}\label{Sect:conclusion}
{How can we find a sampling pattern on the sphere and the rotation group that is also suitable for compressed sensing of signals? By proving \ac{RIP} property, we show that, as it is expected, random sampling patterns can provably be used for signal recovery on the rotation group. The obtained bound depends on the ambient dimension. Future works can focus on improving this dependency by using a change-of-measure similar to \cite{burq2012weighted}. It is currently not clear how the framework of \cite{burq2012weighted} can be adapted for Wigner D-functions. It is interesting to see if the bounds can be improved to only include  logarithmic and poly-logarithmic dependencies on $N$.}

{Given the interest in regular sampling patterns in many applications, we consider various existing regular patterns as well. Interestingly, many patterns with symmetric structure on azimuth and polarization  suffer from high mutual coherence and are essentially unsuitable for compressed sensing. Instead, we propose a new sampling pattern that imposes regularity on elevation. Using tools from angular momentum analysis in quantum mechanics, we show how appropriate elevation sampling patterns can yield mutually incoherent measurements. We show that it is possible to match the lower bound on the coherence for the sphere using a simple coherence minimization algorithm. The phase transition diagrams show that our proposed sampling patterns outperform other regular patterns and surpass even random sampling patterns.} 

{Future works can focus on closing the gap, for the rotation group, between the lower bound and the proposed sampling pattern. This can be done either by deriving new lower bounds or by more effective optimization approaches. 
Another line of research can focus on RIP-free recovery guarantees applicable to deterministic patterns.  We have numerically shown, by using several well-known recovery algorithms, that the proposed sampling points perform better recovery than random as well as the popular regular sampling points. However, the uniform recovery guarantees for the deterministic sampling points suffer from the quadratic bottleneck. Certain works already exist that use number-theoretic construction of \cite{bourgain_explicit_2011} for a deterministic sensing matrix. The extension of these methods to 
 $\S^2$ and $\SO(3)$ is an interesting and non-trivial problem. }

\section*{Acknowledgment}
This work is funded by DFG project (CoSSTra-MA1184 $|$ 31-1).

\appendix
\section{Proof of Theorem  \ref{thm:theorem_wigner}}

\label{proof:thm:theorem_wigner}

We have seen that Wigner d-functions are indeed weighted Jacobi polynomials. An upper bound on general weighted orthonormal functions is discussed in \cite[Theorem 6.1]{rauhut2012sparse} and also in \cite{szeg1939orthogonal}. However, we use directly the upper bound on Wigner d-functions obtained in \cite[Theorem 1.1]{haagerup2014inequalities}.

\begin{lemma} [Bound for Jacobi polynomials Wigner d-functions \cite{haagerup2014inequalities}] \label{wignersmallbound}
For Jacobi polynomials $P_{\alpha}^{(\xi,\lambda)}$ of degree $\alpha$ and of order $(\xi,\lambda) $, there exists a constant $C\geq 0$ such that:
\begin{align}
\bigg|(\sin \theta)^{1/2}  \sqrt{\gamma} \sin^{\xi}\bigg(\frac{\theta}{2}\bigg)& \cos^{\lambda}\bigg(\frac{\theta}{2}\bigg) P_{\alpha}^{(\xi,\lambda)}(\cos \theta)\bigg|\nonumber\\
&\leq C(2\alpha+\xi+\lambda+1)^{-1/4}.
\end{align} 
\end{lemma}
\begin{corollary} [Bound for Wigner d-functions] For Wigner d-functions $\mathrm{d}_l^{k,n}(\cos \theta)$, there exists a constant $C\geq 0$ such that $ \card{(\sin \theta)^{1/2} \mathrm d_l^{k,n}(\cos \theta)} \leq C(2l +1)^{-1/4}.$
\label{corol:upperbound}
\end{corollary}
The previous corollary is easily obtained using $\xi, \lambda \geq 0 $ defined as in Definition \ref{def:WigD} and observing that $2\alpha+\xi+\lambda$ equals $2l$. We will later use this corollary to find an upper bound on  weighted Wigner D-functions. Since Wigner D-functions are orthonormal, it suffices to find a useful upper bound $K$ on them and then using it in Theorem \ref{thm:RIP_BOS}. The following proposition serves this purpose.
\begin{proposition}[Bounds on preconditioned Wigner D-functions] \label{prep_wigner}

The Wigner D-functions $\mathrm D_l^{k,n}(\theta,\phi,\chi)$ preconditioned with $(\sin\theta)^{1/2}$ are an orthonormal basis with respect to the product measure $\mathrm d\nu=\mathrm d\theta \mathrm d\phi \mathrm d\chi$ and satisfy the following upper bound: 
\begin{equation*}
\begin{aligned}
& \underset{\begin{subarray}{c}
   0\leq l\leq B-1\\
   k,n\in\{-l,\dots,l\} 
  \end{subarray}}{\textnormal{sup}} 
& & \left\Vert (\sin\theta)^{1/2} \mathrm D_l^{k,n}(\theta,\phi,\chi) \right\Vert_{\infty}  \leq C_0 N^{\frac{1}{12}},
\end{aligned}
\end{equation*}
where $N$ is the total number of Wigner D-functions of degree less than $B$.
\end{proposition}
\begin{proof}
Using Corollary \ref{corol:upperbound}, we can see that :

\begin{align*}
&\norm{(\sin \theta)^{1/2} N_l \,\D l{k}{n}(\theta,\phi,\chi)}_\infty  = \norm{(\sin \theta)^{1/2} N_l\,d_l^{k,n}(\cos \theta)}_\infty\\
&\quad \leq C\,N_l\,(2l +1)^{-1/4} = \frac{C}{\sqrt{8\pi^2}}(2l+1)^{1/4}\\
&\quad\leq  \frac{C}{\sqrt{8\pi^2}}(2B-1)^{1/4}
\end{align*}
Note that the number of all orthonormal basis functions $N$ is related $B$ by $N =\frac{B(2B-1)(2B+1)}{3}$. Using  the inequality $(2B-1)^3\leq 6N$, we have for some constant $C_0$:
\begin{align*}
 \norm{(\sin \theta)^{1/2} N_l \,\mathrm D_l^{k,n}(\theta,\phi,\chi)}_\infty &\leq  \frac{C}{\sqrt{8\pi^2}}(6N)^{1/12}=C_0 N^{1/12}.
\end{align*}
\end{proof} 

From Proposition \ref{prep_wigner}, we can use Theorem \ref{thm:RIP_BOS} and \ref{thm:RIP_recovery} to prove sparse recovery guarantees for the coefficients of Wigner D-expansion using random samples of the function. 
Consider the functions $\varphi_l^{k,n}(\theta,\phi,\chi)= P(\theta)D_l^{k,n}(\theta,\phi,\chi)$, with product measure $\mathrm d\nu$. Note that the product measure $\mathrm d\nu=\mathrm  d\theta \mathrm  d\phi \mathrm d\chi$ with preconditioning function $P(\theta)^2=\sin(\theta)$ yields the uniform measure. Orthonormality can then be checked easily:
\begin{align*}
&\int_{\mathrm{SO}(3)}\varphi_l^{k,n}(\theta,\phi,\chi)\overline{\varphi_{l'}^{k',n'}\big(\theta,\phi,\chi\big)}\mathrm d\nu\\
&=\int_{\mathrm{SO}(3)}\mathrm D_l^{k,n}(\theta,\phi,\chi)\overline{\mathrm D_{l'}^{k',n'}\big(\theta,\phi,\chi\big)}\sin(\theta)\mathrm  d\theta \mathrm  d\phi \mathrm d\chi\\
&=\delta_{nn'}\delta_{kk'}\delta_{ll'}.
\end{align*}
Therefore the functions $\varphi_l^{k,n}(\theta,\phi,\chi)$ form an orthonormal basis with the upper bound provided in the Proposition \ref{prep_wigner}. Using these with Theorem \ref{thm:RIP_BOS} and \ref{thm:RIP_recovery} finishes the proof.

\bibliographystyle{IEEEtran}
\bibliography{collection}

\end{document}